\documentclass[11pt, a4paper]{amsart}
\usepackage[dvips]{epsfig}
\usepackage{amsmath}
\usepackage{amssymb}
\usepackage{amsfonts}
\usepackage{amsthm}
\usepackage{amsbsy}
\usepackage{amsgen}
\usepackage{amscd}
\usepackage{amsopn}
\usepackage{amstext}
\usepackage{amsxtra}
\usepackage{enumerate}
\usepackage{dsfont}
\usepackage{hyperref}
\usepackage{lineno}
\usepackage{marginnote}
\usepackage{verbatim}

\usepackage{times, graphicx}
\usepackage{eso-pic}
\usepackage{lastpage}

\numberwithin{equation}{section}

\newtheorem{theorem}{Theorem}[section]
\newtheorem{proposition}[theorem]{Proposition}
\newtheorem{lemma}[theorem]{Lemma}
\newtheorem{corollary}[theorem]{Corollary}
\theoremstyle{definition}
\newtheorem{definition}[theorem]{Definition}
\newtheorem{notation}[theorem]{Notation}

\newcommand{\vol}{\operatorname{Vol}}

\newcommand{\Dc}{\mathcal{D}}

\newcommand{\Sc}{\mathcal{S}}
\newcommand{\Nc}{\mathcal{N}}
\newcommand{\Ec}{\mathcal{E}}
\newcommand{\tr}{\operatorname{tr}}
\newcommand{\Bc}{\mathcal{B}}
\newcommand{\Rc}{\mathcal{R}}

\newcommand{\meas}{\operatorname{meas}}

\newcommand {\E} {\mathbb{E}}
\newcommand {\M} {\mathcal{M}}
\newcommand {\R} {\mathbb{R}}
\newcommand {\Cb} {\mathbb{C}}
\newcommand {\Z} {\mathbb{Z}}

\newcommand {\Cc} {\mathcal{C}}
\newcommand {\Lc} {\mathcal{L}}

\newcommand {\Ac} {\mathcal{A}}

\newcommand {\Tb} {\mathbb{T}}
\newcommand {\var} {\operatorname{Var}}
\newcommand {\len} {\operatorname{len}}
\newcommand {\cov} {\operatorname{Cov}}
\newcommand {\corr} {\operatorname{Corr}}
\newcommand{\C}{{\mathbb C}}

\newcommand{\new}[1]{{#1}}
\newcommand{\qnew}[1]{{#1}}
\newcommand{\qcnew}[1]{{#1}}

\begin{document}

\title{Planck-scale distribution of nodal length of arithmetic random waves}

\author{Jacques Benatar}
\address{Department of Mathematics, King's College London, UK,
currently at School of Mathematical Sciences\\Tel Aviv University\\ Tel Aviv 69978\\Israel}
\email{benatar@mail.tau.ac.il}

\author{Domenico Marinucci}
\address{Department of Mathematics, University of Rome ``Tor Vergata", Italy}
\email{marinucc@axp.mat.uniroma.it}

\author{Igor Wigman}
\address{Department of Mathematics, King's College London, UK}
\email{igor.wigman@kcl.ac.uk}

\begin{abstract}
We study the nodal length of random toral Laplace eigenfunctions (``arithmetic random waves")
restricted to decreasing domains (``shrinking balls"), all the way down to Planck scale. We find that, up
to a natural scaling, for ``generic" energies the variance of the restricted nodal length obeys the same asymptotic law as the
total nodal length, and these are asymptotically fully correlated. This, among other things, allows for a statistical reconstruction
of the full toral length based on partial information. One of the key novel ingredients of our work, borrowing from
number theory, is the use of bounds for the so-called spectral Quasi-Correlations, i.e. unusually small sums of lattice
points lying on the same circle.

\end{abstract}

\date{\today}
\maketitle

\section{Introduction}

\subsection{Laplace eigenfunctions in Planck (microscopic) scale}

\label{sec:Lap eig Planck}

Let $(\M,g)$ be a compact smooth Riemannian surface, and $\Delta$ the Laplace-Beltrami operator on $\M$.
It is well-known that in this situation the spectrum of $\Delta$ is purely discrete, i.e. there exists a non-decreasing
sequence $\{ E_{j}\}_{j\ge 1}\subseteq \Z_{\ge 0}$ of eigenvalues of $\Delta$ (``energy levels of $\M$") and the corresponding
eigenfunctions $\{\phi_{j}\}_{j\ge 1}$ such that
\begin{equation}
\label{eq:Deltaf+lambda f=0}
\Delta \phi_{j} + E_{j}\cdot \phi_{j} = 0,
\end{equation}
and $\{\phi_{j}\}_{j\ge 1}$ is an o.n.b. of $L^{2}(\M)$.

We are interested in the {\em nodal line} of $\phi_{j}$, i.e. its zero set $\phi_{j}^{-1}(0)$ in the high-energy limit
$j\rightarrow\infty$; it is generically a smooth curve ~\cite{Uh}. In this general setup Yau's conjecture asserts that
there exist constants $0<c_{\M}<C_{\M}<\infty$ such that the {\em nodal length} (i.e. the length of $\phi_{j}^{-1}(0)$)
satisfies
\begin{equation}
\label{eq:Yau's conj}
c_{\M}\cdot \new{\sqrt{E_{j}}}\ll \len(\varphi_{j}^{-1}(0)) \ll  C_{\M}\cdot \new{\sqrt{E_{j}}}.
\end{equation}
Yau's conjecture was resolved by Donnelly-Fefferman ~\cite{Do-Fe} for $\M$ real analytic, and the lower bound was established
more recently by Logunov ~\cite{Lo1,Lo2,Lo-Ma} for the more general smooth case.

Berry's seminal and widely believed conjecture ~\cite{Be77,Be83} asserts that, at least in some generic situation, one could model the high energy eigenfunctions $\phi_{j}$ on a {\em chaotic} surface $\M$ with random monochromatic plane waves of wavelength $\sqrt{E_{j}}$, i.e. an isotropic random field on $\R^{2}$ with covariance function $$r_{RWM}(x) = J_{0}(\|x\|).$$ When valid, Berry's RWM in particular implies Yau's
conjecture \eqref{eq:Yau's conj}; it goes far beyond the macroscopic setting, i.e. that the RWM is applicable to shrinking domains. For example, it asserts that the RWM is a good model for $\phi_{j}^{-1}(0)\cap B_{r}(x)$, i.e. the nodal length lying inside a shrinking geodesic ball $B_{r}(x)\subseteq \M$, of radius slightly above the Planck scale: $r\approx \frac{C}{\sqrt{E}}$, with $C\gg 0$ sufficiently big.
In this spirit Nadirashvili's conjecture ~\cite{Na} refines upon \eqref{eq:Yau's conj} in that the analogous statement is to hold in the Planck (or microscopic) scale; it was in part established by Logunov ~\cite{Lo1,Lo2}.

To our best knowledge the only other few small-scale results all concern the mass equidistribution. These are some small-scale refinements
~\cite{He-Ri1,He-Ri2,Ha1,Ha2}
for Shnirelman's Theorem \cite{Sn, Ze, CdV} asserting the $L^{2}$ mass equidistribution of $\phi_{j}$ on macroscopic scale for $\M$ chaotic, i.e. that the $L^{2}$ mass of $\phi_{j}$ on a subdomain of $\M$ is proportional to its area along a density $1$ sequence of $\{E_{j}\}$.
For the particular case of the standard flat torus $\M=\R^{2}/\Z^{2}$
Lester-Rudnick ~\cite{Le-Ru} and subsequently Granville-Wigman ~\cite{Gr-Wi} used the number theoretic structure of the toral eigenfunctions
in order to obtain the Planck-scale mass equidistribution or slightly above it for ``most" of the eigenfunctions.

\subsection{Arithmetic Random Waves}

Let
\begin{equation}
\label{eq:S=a^2+b^2}
S= \{ a^{2}+b^{2}:\: a,b\in\Z\}\subseteq \Z
\end{equation}
be the set of integers expressible as a sum of two squares.
For $n\in S$ let $$\Ec_{n}=\{\lambda\in \Z^{2}:\: \|\lambda\|^{2}=n\}$$ be the set of lattice points lying on the radius-$\sqrt{n}$ circle,
and denote its size
\begin{equation*}
\Nc_{n}=r_{2}(n) = |\Ec_{n}|;
\end{equation*}
it is the number $r_{2}(n)$ of ways to express $n$ as sum of two squares. It is well known that every (complex valued) Laplace eigenfunction \eqref{eq:Deltaf+lambda f=0} on the standard $2$-torus $\M=\Tb^{2}=\R^{2}/\Z^{2}$ is necessarily of the form
\begin{equation}
\label{eq:Tn sum def}
T_{n}(x) = \frac{1}{\sqrt{2\Nc_{n}}} \sum\limits_{\lambda\in\Ec_{n}} a_{\lambda}\cdot e(\langle\lambda, x \rangle)
\end{equation}
for some $n\in S$, and coefficients $a_{\lambda}\in\Cb$ (with the rationale of the normalising factor $\frac{1}{\sqrt{2\Nc_{n}}}$
becoming apparent below);
the corresponding eigenvalue is
\begin{equation}
\label{eq:En=4pi^2n}
E_{n}=4\pi^{2}n,
\end{equation}
and $T_{n}$ is real-valued if and only if for every $\lambda\in\Ec_{n}$ we have
\begin{equation}
\label{eq:a(-lambda)=conj(a(lambda)) real}
a_{-\lambda}=\overline{a_{\lambda}}
\end{equation}
(the complex conjugate of $a_{\lambda}$).

For $n\in S$ fixed, the eigenspace of all functions \eqref{eq:Tn sum def} satisfying $$\Delta T_{n}+E_{n}\cdot T_{n}=0$$
could be endowed with a Gaussian probability measure by assuming that the coefficients $\{a_{\lambda}\}_{\lambda\in\Ec_{n}}$ are
standard (complex) Gaussian i.i.d. save for \eqref{eq:a(-lambda)=conj(a(lambda)) real}. With a slight abuse of notation,
the resulting {\em random field}, also denoted $T_{n}$, is the wavenumber-$\sqrt{n}$ ``arithmetic random wave" ~\cite{O-R-W,K-K-W}. Alternatively,
$T_{n}$ is the (unique) centred Gaussian stationary random field with the covariance function
\begin{equation}
\label{eq:rn covar def}
r(x-y)=r_{n}(x-y) =\E[T_{n}(x)\cdot T_{n}(y)]  = \frac{1}{\Nc_{n}}\sum\limits_{\lambda\in\Ec_{n}} e(\langle \lambda, x-y\rangle)=\frac{1}{\Nc_{n}}\sum\limits_{\new{\lambda}\in\Ec_{n}} \cos(2\pi \langle \lambda, x-y\rangle);
\end{equation}
as $r(0)= 1$, the field $T_{n}$ is unit variance (this is the reason we set \eqref{eq:Tn sum def} to
normalise $T_{n}$ in the first place). The (random) zero set $T_{n}^{-1}(0)$ is of our fundamental interest;
it is a.s. a smooth curve ~\cite{R-W2008}, called the {\em nodal line}. Of our particular interest is the distribution of
its (random) total length $$\Lc_{n} = \len\left( T_{n}^{-1}(0)\right),$$
or the length constrained inside subdomains,
$$\Lc_{n;s} = \len\left( T_{n}^{-1}(0) \cap B(s) \right),$$
where $B(s)\subseteq \Tb^{2}$ is a radius-$s$ ball shrinking at Planck scale rate $s>n^{-1/2}$, or, more realistically,
slightly above it $s>n^{-1/2+\epsilon}$.
To be able to explain the context and formulate our main results we will require some background on the arithmetic of lattice points $\Ec_{n}$
lying on circles.

\subsection{Some arithmetic aspects of lattice points $\Ec_{n}$}

Recall that $S\subseteq \Z$ as in \eqref{eq:S=a^2+b^2} is the set of integers expressible as a sum of two squares, and given $n\in S$, the number of such expressions is $\Nc_{n}=r_{2}(n)$. As the distribution of the above central quantities will depend on both $\Nc_{n}$ and the angular distribution of lattice points lying on the corresponding circle, here we give some necessary background.
First, it is known ~\cite{Lan} that $\Nc_{n}$ grows {\em on average} as $$\Nc_{n}\sim c_{0}\cdot \sqrt{\log{n}}$$ with some $c_{0}>0$;
equivalently, as $X\rightarrow\infty$,
\begin{equation}
\label{eq:S(x) asympt cRL x/log(x)}
S(X):=|\{n\in S:\: n \le X   \}  | \sim c_{LR} \cdot \frac{X}{\sqrt{\log{X}}}
\end{equation}
where $c_{LR}>0$ is the (fairly explicit) Ramanujan-Landau constant.
Moreover, $$\Nc_{n} \sim (\log{n})^{(\log{2})/2 + o(1)}$$ for a density $1$ sequence of numbers $n\in S' \subseteq S$
(though we bear in mind that $S\subseteq \Z$ itself is thin or density $0$). To the other end, $\Nc_{n}$
is as small as $\Nc_{p}=8$ for an infinite
sequence of primes $$p\equiv 1 \mod {4},$$ and, in general, it is subject to large and erratic fluctuation satisfying
for every $\epsilon>0$
\begin{equation}
\label{eq:Nc dim << n^eps}
\Nc_{n} = O(n^{\epsilon}).
\end{equation}
From this point on we will always work with (generic) subsequences $\{n\}=S'\subseteq S$ satisfying $\Nc_{n}\rightarrow\infty$.

To understand the angular distribution of the lattice points $\Ec_{n}$ we define the probability measures $\tau_{n}$ on the unit circle $\Sc^{1}\subseteq \R^{2}$:
\begin{equation}
\label{eq:taun atomic def}
\tau_{n}=\frac{1}{\Nc_{n}}\sum\limits_{\lambda\in\Ec_{n}}\delta_{\lambda/\sqrt{n}}.
\end{equation}
It is known ~\cite{Erdos-Hall,Kurlberg-F-W} that for a ``generic" (density $1$) sequence $\{n\}\subseteq S$ the lattice points $\Ec_{n}$ are {\em equidistributed} in the sense that $\tau_{n} \Rightarrow \frac{d\theta}{2\pi},$ i.e. weak-$*$ convergence of probability measures to the uniform measure on $\Sc^{1}$ parameterized as $(\cos{\theta},\sin{\theta})$. To the other extreme, there exist ~\cite{Ci} (thin) {\em ``Cilleruelo"} sequences $\{n\}\subseteq S$ such that the number of lattice points $\Nc_{n}\rightarrow\infty$ grows, though all of them are concentrated
$$\tau_{n} \Rightarrow \frac{1}{4}\left( \delta_{\pm 1} + \delta_{\pm i}  \right)$$ around the four points $\pm 1, \pm i$ where we are thinking of $\Sc^{1}\subseteq \Cb$ as embedded inside the complex numbers.
There exist ~\cite{K-K-W,K-W} other {\em attainable} measures $\tau$ on $\Sc^{1}$, i.e. weak-$*$ partial limits of the sequence
$\{\tau_{n}\}_{n\in S}$; by the compactness of $\Sc^{1}$ the limit measure $\tau$ is automatically a probability measure; a partial classification
of such $\tau$ was obtained ~\cite{K-W} via their Fourier coefficients. In particular it follows that the $4$th Fourier coefficient of attainable measures is unrestricted: for every $\eta \in [-1,1]$ there exists an attainable measure $\tau$
with $\widehat{\tau}(4) = \eta$.

\subsection{Nodal length}

Recall that $$\Lc_{n} = \len\left( T_{n}^{-1}(0)\right)$$ is the total {\em nodal length} of $T_{n}$, and for $0<s<1/2$ (say),
$$\Lc_{n;s} = \len\left( T_{n}^{-1}(0) \cap B(s) \right)$$ is the nodal length of $T_{n}$ restricted to a radius-$s$ ball $B(s)$,
where by the stationarity of $T_{n}$ we may assume that $B(s)$ is centred. A straightforward computation \new{~\cite[Proposition 4.1]{R-W2008}}
with the Kac-Rice formula was used to evaluate the expected length
\begin{equation*}
\E[\Lc_{n}] = \frac{1}{2\sqrt{2}} \new{\sqrt{E_{n}}},
\end{equation*}
consistent with Yau's conjecture, \new{and, by the stationarity of $T_{n}$, the more general result
\begin{equation*}
\E[\Lc_{n;s}] = \frac{1}{2\sqrt{2}} (\pi s^{2})\cdot \sqrt{E_{n}}
\end{equation*}
for the restricted length also follows from the same computation}.

The asymptotic behaviour of the variance of $\Lc_{n}$, eventually resolved by ~\cite{K-K-W}, is of a far more delicate nature (and even more so of $\Lc_{n;s}$); it turned out that it is intimately related to the angular distribution \eqref{eq:taun atomic def} of $\Ec_{n}$.
More precisely it was found ~\cite{K-K-W} that\footnote{Originally only $o(1)$ for the error term claimed. It is easy
to obtain the $O\left(\frac{1}{\Nc_{n}^{1/2}}\right)$ bound for the error term using Bombieri-Bourgain's bound for the length-$6$ spectral
correlations (see \S\ref{sec:spec corr (quasi)}).}
\begin{equation}
\label{eq:var(tot nod length)}
\var(\Lc_{n}) = c_{n}\cdot \frac{E_{n}}{\Nc_{n}^{2}}\cdot \left(1+ O\left(\frac{1}{\Nc_{n}^{1/2}}\right) \right),
\end{equation}
where the leading coefficients
\begin{equation}
\label{eq:cn KKW def}
c_{n} = \frac{1+\widehat{\tau_{n}}(4)^{2}}{512} \in \left[ \frac{1}{512},\frac{1}{256}  \right],
\end{equation}
depending on the arithmetics of $\Ec_{n}$, are bounded away from both $0$ and $\infty$.
The asymptotic formula \eqref{eq:var(tot nod length)} for the nodal length variance shows that in order for
$\var(\Lc_{n})$ to observe an asymptotic law it is essential to separate $S$ into sequences $\{n\}\subseteq S$ such
that the corresponding $\tau_{n}\Rightarrow\tau$ for some (attainable) $\tau$; in this case
$$\var(\Lc_{n}) \sim c(\tau)\cdot \frac{E_{n}}{\Nc_{n}^{2}}$$ with
$$c(\tau) = \frac{1+\widehat{\tau}(4)^{2}}{512}.$$ The variance \eqref{eq:var(tot nod length)} is significantly smaller
compared to the previously expected ~\cite{R-W2008} order of magnitude $ \approx \frac{E_{n}}{\Nc_{n}}$;
it is an arithmetic manifestation of ``Berry's cancellation" ~\cite{Berry2002}, also interpreted ~\cite{MaWi1,MaWi2} as the precise vanishing
of the second chaotic component in the Wiener chaos expansion of $\Lc_{n}$ (or its spherical analogue).
One might expect Berry's cancellation to be a feature of the symmetries
of the full torus; that this is not so follows in particular from a principal result of the present manuscript, Theorem
\ref{thm:toral univ gen} below (see \eqref{eq:Var asympt gen} and \S\ref{sec:disc Berry cancel} for more details).

\vspace{2mm}

The fine distribution of $\Lc_{n}$ was also investigated ~\cite{M-P-R-W}.
For a number $\eta\in [0,1]$ let $M_{\eta}$ be the random variable
\begin{equation*}
M_{\eta} :=  \frac{1}{\sqrt{1+\eta^{2}}} \cdot \left( 2-  (1+\eta)X_{1}^{2} -(1-\eta)X_{2}^{2}    \right),
\end{equation*}
where $(X_{1},X_{2})$ are standard Gaussian i.i.d.; for example, if $\eta=0$, the distribution of $M_{\eta}$ is
a linear transformation of $\chi^{2}$ with $2$ degrees. It was shown ~\cite{M-P-R-W} that
as $\Nc_{n}\rightarrow\infty$, the distribution law of the normalised $$\widetilde{\Lc_{n}}:=\frac{\Lc_{n}-\E[\Lc_{n}]}{\var(\Lc_{n})}$$
is asymptotic to that of $M_{|\widehat{\tau_{n}}(4)|}$, both in mean\footnote{In fact, in $L^{p}$ for all $p\in (0,2)$.}
and almost surely. The meaning of the convergence in mean is that
there exist copies of $\widetilde{\Lc_{n}}$ and $M_{|\widehat{\tau_{n}}(4)|}$, defined on the
same probability space, such that
\begin{equation}
\label{eq:tot length NCLT}
\E\left[|\widetilde{\Lc_{n}} - M_{|\widehat{\tau_{n}}(4)|} |\right] \rightarrow 0
\end{equation}
as $\Nc_{n}\rightarrow\infty$.

The first principal result of this manuscript is that an analogous statement to \eqref{eq:var(tot nod length)}
holds for the nodal length $\Lc_{n;s}$ of $T_{n}$ restricted to shrinking balls slightly above
Planck scale, for {\em generic} sequences of energy levels $\{ n\}\subseteq S$, and that in this regime
$\Lc_{n;s}$ are asymptotically fully correlated with $\Lc_{n}$; this also implies that an analogue of \eqref{eq:tot length NCLT}
holds for $\Lc_{n;s}$ (see Corollary \ref{cor:shrink length Meta}). Below we will specify a generic arithmetic condition, sufficient for
the conclusions of Theorem \ref{thm:toral univ gen} to hold (see Theorem \ref{thm:toral univ spec corr}).

\begin{theorem}
\label{thm:toral univ gen}
For every $\epsilon>0$ there exists a density-$1$ sequence of numbers $$S'=S'(\epsilon)\subseteq S$$ so that the following hold.

\begin{enumerate}

\item Along $n\in S'$ we have $\Nc_{n}\rightarrow \infty$, and the set of accumulation
points of $\{\new{\widehat{\tau_{n}}(4)}\}_{n\in S'}$ contains the interval $[0,1]$.

\item For $n\in S'$, uniformly for all $s>n^{-1/2+\epsilon}$ we have
\begin{equation}
\label{eq:Var asympt gen}
\var( \Lc_{n;s}) = c_{n} \cdot (\pi s^{2})^{2}\cdot \frac{E_{n}}{\Nc_{n}^{2}}
\left(1+O_{\epsilon}\left(\frac{1}{\Nc_{n}^{1/2}} \right) \right),
\end{equation}
where $c_{n}$ is given by \eqref{eq:cn KKW def}, and the constant involved in the `$O$'-notation depends on $\epsilon$ only
(cf. \eqref{eq:var(tot nod length)}).

\item
For random variables $X,Y$ we denote as usual their correlation $$\corr(X,Y) := \frac{\cov(X,Y)}{\sqrt{\var(X)}\cdot \sqrt{\var(Y)}}.$$
Then for every $\epsilon>0$ we have that
\begin{equation}
\label{eq:restr nod len full corr full}
\sup\limits_{s>n^{-1/2+\epsilon}}\left| \corr(\Lc_{n;s},\Lc_{n})-1\right| \rightarrow 0,
\end{equation}
i.e. the nodal length $\Lc_{n;s}$ of $f_{n}$ restricted to a small ball is asymptotically fully correlated with the full nodal length
$\Lc_{n}$ of $f_{n}$, uniformly for all $s>n^{-1/2+\epsilon}$.

\end{enumerate}

\end{theorem}

Quite remarkably (and somewhat surprisingly), \eqref{eq:restr nod len full corr full} shows that one may statistically reconstruct
the full nodal length of $f_{n}$ based on its restriction to a small toral ball.
Since the full nodal length $\Lc_{n}$ of $f_{n}$ obeys the limiting law \eqref{eq:tot length NCLT}, the full correlation
\eqref{eq:restr nod len full corr full} of the restricted nodal length $\Lc_{n;s}$ with $\Lc_{n}$ implies
the same limiting law for
$\Lc_{n;s}$ under the same conditions as Theorem \ref{thm:toral univ gen}. More precisely we have the following corollary:

\begin{corollary}
\label{cor:shrink length Meta}

Given $\epsilon>0$ and a (generic) sequence $S'=S'(\epsilon)\subseteq S$ as in Theorem \ref{thm:toral univ gen},
there exists a coupling $(M_{|\new{\widehat{\tau_{n}}(4)}|},f_{n})$ (of a random variable with a random field) satisfying
\begin{equation*}
\sup\limits_{s>n^{-1/2+\epsilon}}\E\left[\left| \frac{\Lc_{n;s}-\E[\Lc_{n;s}]}{\var(\Lc_{n;s})}  -  M_{|\new{\widehat{\tau_{n}}}(4)|}   \right|\right]
\rightarrow 0.
\end{equation*}

\end{corollary}

\vspace{3mm}

\subsection{Spectral (Quasi)-Correlations}
\label{sec:spec corr (quasi)}

Our next goal is to formulate a result \`{a} la Theorem \ref{thm:toral univ gen} with a more explicit control over $\{n\}$ satisfying the
conclusions of Theorem \ref{thm:toral univ gen}. To this end we will require some more notation.
Recall that the covariance function $r_{n}(x)$ of $T_{n}$ is given by \eqref{eq:rn covar def}; of our particular interest are the {\em moments}
of $r_{n}$ (and related quantities), both on the whole of $\Tb^{2}$ and restricted to $B(s)$. More precisely, for $l\ge 1$ we define
the ``full" moments of $r_{n}$ as
\begin{equation}
\Rc_{n}(l) =\int\limits_{\Tb^{2}\times\Tb^{2}}r_{n}(x-y)^{l}dxdy=\int\limits_{\Tb^{2}}r_{n}(x)^{l}dx
\end{equation}
by the stationarity, and the ``{\em restricted moments}" of $r_{n}$ as
\begin{equation}
\label{eq:Rc moments rn def}
\Rc_{n}(l;s) =\int\limits_{B(s)\times B(s)}r_{n}(x-y)^{l}dxdy.
\end{equation}
An explicit computation using the orthogonality relations
\begin{equation}
\label{eq:orthog rel exp int}
\int\limits_{\Tb^{2}}e(\langle x,\xi\rangle)dx = \begin{cases}
1 &\xi=0 \\ 0 &\xi\ne 0
\end{cases},
\end{equation}
$\xi\in\Z^{2}$, relates $\Rc_{n}(l)$ to the set of length-$l$ spectral correlations
\begin{equation}
\label{eq:Sc correlations def}
\Sc_{n}(l) = \left\{ (\lambda_{1},\ldots,\lambda_{l})\in(\Ec_{n})^{l}:\: \sum\limits_{i=1}^{l}\lambda_{i}=0  \right\};
\end{equation}
the full moments of $r_{n}$ are given by
\begin{equation*}
\Rc_{n}(l) = \frac{|\Sc_{n}(l)|}{\Nc_{n}^{l}}.
\end{equation*}

For $l=2k$ even we further define the {\em diagonal} spectral correlations set to be
\begin{equation}
\label{eq:Dc diag def}
\Dc_{n}(l) = \left\{ \pi(\lambda_{1},-\lambda_{1},\ldots,\lambda_{k},-\lambda_{k}): \lambda_{1},\ldots,\lambda_{k}\in (\Ec_{n})^{k},\,\pi\in S_{l}\right\}
\end{equation}
the set of all possible permutations of tuples of lattice points of the form $(\lambda_{1},-\lambda_{1},\ldots,\lambda_{k},-\lambda_{k})$.
It is evident that in this case $$\Dc_{n}(l) \subseteq \Sc_{n}(l),$$ and that for every fixed $l$ we have
$$|\Dc_{n}(l)|=  \frac{(2k)!}{2^{k}\cdot k!}\Nc_{n}^{k}\cdot \left(1 + O_{\Nc_{n}\rightarrow\infty}\left( \frac{1}{\Nc_{n}} \right)  \right);$$
hence for every $k\ge 1$ we have the lower bound $$|\Sc_{n}(2k)| \gg \Nc_{n}^{k}.$$ To the other end, for $k=1$ we have
$\Sc_{n}(2)=\Dc_{n}(2),$ by the definition, whereas for $k=2$ the equality
\begin{equation}
\label{eq:Dc4=Sc4}
\Dc_{n}(4) = \Sc_{n}(4),
\end{equation}
is due to an elegant (and simple)
geometric observation by Zygmund ~\cite{Zygmund} (``Zygmund's trick"); the same observation yields $$\Sc_{n}(6)\ll \Nc_{n}^{4}.$$

A key ingredient in ~\cite{K-K-W} was a non-trivial improvement for the latter bound,
\begin{equation}
\label{eq:BB Sn(6)<<N^7/2}
|\Sc_{n}(6)|= o(\Nc_{n}^{4})
\end{equation}
due to Bourgain (published in ~\cite{K-K-W}) implying that the l.h.s. of \eqref{eq:var(tot nod length)} is asymptotic to the r.h.s
of \eqref{eq:var(tot nod length)} (though not implying the stated bound for the error term); a further improvement
\begin{equation}
\label{eq:S6<<N^7/2}
\Sc_{n}(6)\ll \Nc_{n}^{7/2}
\end{equation}
\new{holding ~\cite{B-B} for the full sequence
$n\in S$} yields the stronger form \eqref{eq:var(tot nod length)} of their result with the prescribed error term. The sharp upper bound
\begin{equation*}
|\Sc_{n}(6)| \ll \Nc_{n}^{3},
\end{equation*}
or, even more striking,
\begin{equation*}
|\Sc_{n}(6)| = 3\Nc_{n}^{3}+O(\Nc_{n}^{3-\delta})
\end{equation*}
(equivalently, $|\Sc_{n}(6)\setminus \Dc_{n}(6)| \ll \Nc_{n}^{3-\delta}$) \new{holds for a density $1$ subsequence $\{n\}\subseteq S$}.

\vspace{3mm}

For the restricted moments of particular interest for our purposes, we need to consider the ``spectral {\em quasi}-correlations"
(see \S\ref{sec:rest mom quasi corr}),
i.e. the defining
equality $\sum\limits_{i=1}^{l}\lambda_{i}=0$ in \eqref{eq:Sc correlations def} holding approximately: the absolute value
$|\cdot | \le K$, where the parameter $K$ is typically of order of magnitude $K\approx n^{1/2-\delta}$, $0<\delta<\epsilon$.

\vspace{3mm}

\begin{definition}[Quasi-correlations and $\delta$-separatedness.]
\label{def:separatedness Ac}

\begin{enumerate}

\item Given a number $n\in S$, $l\in\Z_{\ge 2}$, and $0 < K=K(n) < l\cdot\sqrt{n}$ we define the set of length-$l$ spectral quasi-correlations
\begin{equation}
\label{eq:Cc qcorr def}
\Cc_{n}(l;K) =
\left\{(\lambda_{1},\ldots,\lambda_{l})\in\Ec_{n}^{l}:\: 0 < \left\| \sum\limits_{j=1}^{l}\lambda_{j} \right\| \le K  \right\}.
\end{equation}

\item For $\delta>0$ we say that a number $n\subseteq S$ satisfies the $(l,\delta)$-separatedness hypothesis $\Ac(n;l,\delta)$ if
$$\Cc_{n}(l;n^{1/2-\delta}) = \varnothing.$$

\end{enumerate}

\end{definition}

\vspace{3mm}

For example, a number $n\in S$ satisfies the $(2,\delta)$-separatedness hypothesis $\Ac(n;2,\delta)$ if the nearest neighbour distance in $\Ec_{n}$ grows like $n^{1/2-\delta}$, i.e. for all $\lambda,\lambda'\in\Ec_{n}$ with $\lambda\ne\lambda'$ we have
$$\|\lambda-\lambda'\|> n^{1/2-\delta}.$$ Bourgain and Rudnick ~\cite[Lemma $5$]{B-R11} proved that all but $\new{O(X^{1-2\delta/3})}$ numbers
$$n\in S(X)=\{n\in S:\: n\le X\}$$ satisfy $\Ac(n;2,\delta)$ (cf. \eqref{eq:S(x) asympt cRL x/log(x)}),
and more recently Granville-Wigman ~\cite{Gr-Wi} refined their estimate
to yield a precise asymptotic expression for the number of exceptions $n$ to $\Ac(n;2,\delta)$.
More generally, the following theorem shows that for generic $n\in S$ the assumption
$\Ac(n;l,\delta)$ holds for all $n\ge 2$; it is stronger than needed for our purposes in more than one way
(see \S\ref{sec:rest mom quasi corr}).

\begin{theorem}
\label{thm:quasi-corr small}
For every $l\ge 2$ and $\delta>0$ there exist a set $S'=S'(l;\delta)\subseteq S$ such that:

\begin{enumerate}

\item The set $S'\subseteq S$ has density $1$ in $S$.

\item The set of accumulation points of $\{ \widehat{\new{\tau_{n}}}(4)\}_{n\in S'}$ contains the interval $[0,1]$.

\item For every $n\in S'$ the length-$l$
spectral quasi-correlation set $$\Cc_{n}(l;n^{1/2-\delta})=\varnothing$$ is empty, i.e. $\Ac(n;l,\delta)$ is satisfied.

\end{enumerate}

\end{theorem}

The following result is a version of Theorem \ref{thm:toral univ gen} with an explicit sufficient condition on $n$, by virtue of Theorem
\ref{thm:quasi-corr small}.

\begin{theorem}
\label{thm:toral univ spec corr}

Let $\epsilon>0$, $\delta<\epsilon$, and $S'=\{n\}\subseteq S$ be a sequence of energies so that for all $n\in S'$ the hypotheses
$\Ac(n;2,n^{1/2-\delta/2})$ and $\Ac(n;6,n^{1/2-\delta})$ are satisfied.
Then along $S'$ both \eqref{eq:Var asympt gen} and \eqref{eq:restr nod len full corr full} of Theorem \ref{thm:toral univ gen} hold.

\end{theorem}

In light of Theorem \ref{thm:quasi-corr small}, Theorem \ref{thm:toral univ spec corr} clearly implies the statement of Theorem
\ref{thm:toral univ gen}, so from this point on we will only aim at proving Theorem \ref{thm:toral univ spec corr}.

\subsection*{Acknowledgements}

We are indebted to Jerry Buckley, Manjunath Krishnapur, P\"{a}r Kurlberg, Ze\'{ev} Rudnick and Mikhail Sodin for many stimulating and fruitful conversations, and their valuable comments on the earlier version of this manuscripts. It is a pleasure to thank Andrew Granville for his inspiring ideas, especially related to some aspects involving various subjects in the geometry of numbers. \new{Finally, we are grateful to the anonymous
referee for his very thorough reading of the manuscript, also resulting in improved readability of the present version.}
The research leading to these results has received funding from
the European Research Council under the European Union’s Seventh Framework Programme
(FP7/2007-2013) / ERC grant agreements n$^{\text{o}}$  277742 Pascal (D. M.) and n$^{\text{o}}$  335141 Nodal (J.B. and I.W.),

\section{Discussion}
\label{sec:discussion}

\subsection{Berry's cancellation phenomenon}

\label{sec:disc Berry cancel}

It was originally found by Berry ~\cite{Berry2002} that
the nodal length variance of the RWM (see \S\ref{sec:Lap eig Planck})
in growing domains, e.g. the radius-$R$ Euclidean balls $B(R)\subseteq \R^{2}$ is of lower order than what was expected from
the scaling considerations, i.e. of order of magnitude $\log{R}$, rather than $R$: ``...it results
from a cancellation whose meaning is still obscure..."; it was found that the leading term of the $2$-point correlation
function is purely oscillatory. The same phenomenon (``Berry's Cancellation") was rediscovered ~\cite{Wi09} for
random high degree spherical harmonics, and then for the arithmetic random waves ~\cite{K-K-W} (``Arithmetic Berry's Cancellation").

In ~\cite{M-P-R-W} the Wiener chaos expansion was applied to the nodal length, and it was interpreted that Berry's cancellation has to do with the precise vanishing of the projection of the nodal length into the $2$nd chaos, with the $4$th one dominating. A similar observation
with the $4$th chaotic projection dominating was also made for the high degree $l\rightarrow\infty$ spherical harmonics ~\cite{M-R-W}, also
for shrinking domains of Planck scale (e.g. radius $\frac{R}{l}$ spherical caps with $R=R(l)\rightarrow\infty$ arbitrarily slowly); though
for the shrinking domains the $2$nd chaotic projection {\em does not vanish precisely}, it is still dominated by the $4$th one.

The ability to understand the asymptotic behaviour of the nodal length variance for the torus ~\cite{K-K-W} depended on evaluating the various
moments of the covariance function $r_{n}$ in \eqref{eq:rn covar def} via the orthogonality relations \eqref{eq:orthog rel exp int} holding for the full torus (see \S\ref{sec:rest mom quasi corr} to follow immediately);
this no longer holds for shrinking domains (or even fixed subdomains of $\Tb^{2})$. It was then thought that the analogous results
of ~\cite{K-K-W} and further ~\cite{M-P-R-W} fail decisively for domains shrinking within Planck scale rate; our principal results show
that the contrary is true (Berry's cancellation holding;
the second chaos projection, though not vanishing precisely, being dominated by the $4$th one; understanding the limit distribution
of the nodal length) if we are willing to excise a {\em thin} set of energies $\{n\}\subseteq S$ and work {\em slightly above} the Planck scale.

\subsection{Restricted moments on shrinking domains and spectral quasi-correlations}

\label{sec:rest mom quasi corr}

The principal result of this manuscript asserts that the nodal length distribution for arithmetic random waves \eqref{eq:Tn sum def}
on the whole torus ~\cite{K-K-W,M-P-R-W} is, up to a normalising factor, asymptotic to the nodal length restricted to balls shrinking slightly above
Planck scale, albeit for generic energy levels $n$ only. Let $r_{n}(x,y)=r_{n}(x-y)$ be the covariance function
\eqref{eq:rn covar def} of $T_{n}$; one may expand ~\cite{K-K-W} the nodal length variance in terms of the moments of $r_{n}$ and its derivatives, and these are also intimately related to the finer aspects of its limit distribution \cite{M-P-R-W}. Let us consider the $2$nd moment of $r_{n}$ as an illustrative example; for the unrestricted problem (i.e. the full torus) we have by the translation invariance
\begin{equation}
\label{eq:rn 2nd moment comp full}
\begin{split}
\int\limits_{\Tb^{2}\times\Tb^{2}}r_{n}(x-y)^{2}dxdy &= \int\limits_{\Tb^{2}\times\Tb^{2}}r_{n}(x)^{2}dxdy
= \frac{1}{\Nc_{n}^{2}}\sum\limits_{\lambda,\lambda'\in\Ec_{n}}\int\limits_{\Tb^{2}}e(\langle x,\lambda-\lambda' \rangle)dx
\\&= \frac{1}{\Nc_{n}} + \frac{1}{\Nc_{n}^{2}}\sum\limits_{\lambda\ne\lambda'}\int\limits_{\Tb^{2}}e(\langle x,\lambda-\lambda' \rangle)dx
= \frac{1}{\Nc_{n}},
\end{split}
\end{equation}
upon separating the diagonal, and using the orthogonality relations \eqref{eq:orthog rel exp int}.

For the restricted moments, e.g. $$\int\limits_{B(s)\times B(s)}r_{n}(x-y)^{2}dxdy,$$ an analogue of \eqref{eq:rn 2nd moment comp full}
no longer holds, as we no longer have the precise orthogonal relations \eqref{eq:orthog rel exp int}
nor the translation invariance (first equality in \eqref{eq:rn 2nd moment comp full}). We may still separate the diagonal to write
\begin{equation}
\label{eq:rn 2nd moment comp rest}
\begin{split}
\int\limits_{B(s)\times B(s)}r_{n}(x-y)^{2}dxdy &= (\pi s^{2})^{2}\cdot \frac{1}{\Nc_{n}} + \frac{1}{\Nc_{n}^{2}}\sum\limits_{\lambda\ne\lambda'}\int\limits_{B(s)\times B(s)}e(\langle x-y,\lambda-\lambda' \rangle)dxdy
\\&= (\pi s^{2})^{2}\cdot \frac{1}{\Nc_{n}} +
\frac{1}{\Nc_{n}^{2}}\sum\limits_{\lambda\ne\lambda'}\left|\int\limits_{B(s)}e(\langle x,\lambda-\lambda' \rangle)dx\right|^{2},
\end{split}
\end{equation}
so we no longer need to cope with the lack of translation invariance. Comparing \eqref{eq:rn 2nd moment comp rest} with
\eqref{eq:rn 2nd moment comp full} we observe that both have the same {\em diagonal} contribution, with the off-diagonal one
for \eqref{eq:rn 2nd moment comp rest} might not be vanishing. We then observe that the inner integral on the r.h.s. of
\eqref{eq:rn 2nd moment comp rest} is the Fourier transform of the characteristic function $\chi_{B(s)}$ of the Euclidean
unit disc $B(s)\subseteq \R^{2}$, evaluated at $\lambda-\lambda'$; by scaling, we have that
\begin{equation*}
\int\limits_{B(s)}e(\langle x,\lambda-\lambda' \rangle)dx = s^{2}\chi_{B(1)}(s\cdot \|\lambda-\lambda'\|) = 2\pi s^{2}\cdot \frac{J_{1}(s \|\lambda-\lambda'\|)}{s\|\lambda-\lambda\|}
\end{equation*}
where $J_{1}$ is the Bessel $J$ function (cf. \eqref{eq:chihat(xi)via J1}). We then obtain
\begin{equation}
\label{eq:rn 2nd moment Bessel}
\int\limits_{B(s)\times B(s)}r_{n}(x-y)^{2}dxdy = (\pi s^{2})^{2}\frac{1}{\Nc_{n}} +
2\pi s^{2} \sum\limits_{\lambda\ne \lambda'}\frac{J_{1}(s\| \lambda-\lambda'\|)^{2}}{s^{2}\|\lambda-\lambda'\|^{2}};
\end{equation}
since $J_{1}$ decays at infinity, it is evident that for the diagonal contribution to dominate
the r.h.s. of \eqref{eq:rn 2nd moment Bessel} it is important to control the contribution of the regime
$s\| \lambda-\lambda'\|  \ll 1. $

Since $s$ is assumed to be above the Planck scale $s>n^{-1/2+\epsilon}$,
and $\Nc_{n}$ is much smaller \eqref{eq:Nc dim << n^eps} than any power of $n$, it is sufficient to bound the contribution of the range
\new{$\|\lambda-\lambda'\| < n^{1/2-\delta}$} for some $\delta<\epsilon$, i.e. the size of the quasi-correlation set
\new{$\Cc_{n}(2;n^{1/2-\delta})$}. Recalling the notation \eqref{eq:Rc moments rn def} for restricted moments,
the above discussion shows that, under the assumption $\Ac(n;2,\delta)$ that $\Cc_{n}(l;n^{1/2-\delta})=\emptyset$,
the second moment of $r_{n}$, restricted to $B(s)$, is asymptotic to
\begin{equation}
\label{eq:R2nd mom asymp 1/N}
\Rc_{n}(2;s) \sim (\pi s^{2})^{2}\cdot \frac{1}{\Nc_{n}}.
\end{equation}
Theorem \ref{thm:quasi-corr small} (also the aforementioned Bourgain-Rudnick's ~\cite[Lemma $5$]{B-R11})
shows that the hypothesis $\Ac(n;2,\delta)$ is satisfied for a density $1$ sequence $\{n\}\subseteq S$;
clearly, $\Ac(n;2,\delta)$ not allowing any quasi-correlations is far stronger than what is required for \eqref{eq:R2nd mom asymp 1/N} to be satisfied, by the above.
Instead, in order to yield the asymptotics \eqref{eq:R2nd mom asymp 1/N} for the second restricted moment, it would be sufficient to impose that
the quasi-correlation set $$\Cc_{n}(l;n^{1/2-\delta}) = o(\Nc_{n}) $$ is dominated by the diagonal $\Dc_{n}(2)$.

More generally, for the other relevant moments associated to $r_{n}$ (namely, higher moments of $r_{n}$ or its derivatives, or second moment of various derivatives of $r_{n}$) we need to expand the restricted moments up to an error term $o\left(\frac{1}{\Nc_{n}^{2}}\right)$. As for the nodal length computations we need to evaluate the $2$nd and $4$th moments and bound the $6$th restricted moment $$\Rc_{n}(6;s) = o\left(\frac{s^{4}}{\Nc_{n}^{2}}\right)$$ (see \eqref{eq:L2 approx def} and \eqref{eq:L2 mom ball estimate} below),
that naturally brings up the questions of bounding the quasi-correlation sets $\Cc_{n}(2;n^{1/2-\delta})$,
$\Cc_{n}(4;n^{1/2-\delta})$, and $\Cc_{n}(6;n^{1/2-\delta})$ (see Lemma \ref{lem:r der moments} and its proof below),
with the diagonal contribution coming from \new{$\Sc_{n}(2)=\Dc_{n}(2)$, $\Sc_{n}(4)=\Dc_{n}(4)$ or $\Sc_{n}(6)$ respectively
(and $\Sc_{n}(6)$} being bounded by ~\cite{B-B}).
It then follows that the conclusions of Theorem \ref{thm:toral univ spec corr}
hold uniformly for $s>n^{-1/2+\epsilon}$ under the hypotheses $\Ac(n;2,\delta)$, $\Ac(n;4,\delta)$ and $\Ac(n;6,\delta)$ for some $\delta<\epsilon$;
this is somewhat different from the assumptions made within the formulation of Theorem \ref{thm:toral univ spec corr}. It was observed
~\cite[Lemma $5.2$]{R-W2017} that, in fact, the hypothesis $\Ac(n;2,n^{1/2-\delta})$ implies $\Ac(n;4,n^{1/2-2\delta})$,
explaining the said discrepancy:
\begin{lemma}[~\cite{R-W2017}, Lemma $5.2$]
\label{lem:RW qcorr2=>qcorr4}
For $\delta<1/2$, $n\in S$ sufficiently big, if $n$ satisfies the separatedness hypothesis $\Ac(n;2,n^{1/2-\delta})$, then $n$ also satisfies
$\Ac(n;4,n^{1/2-2\delta})$.
\end{lemma}

By the above, rather than assuming that
$$\Cc_{n}(2;n^{1/2-\delta/2})=\Cc_{n}(6;n^{1/2-\delta})=\varnothing$$ are empty, it would be sufficient to make the somewhat weaker assumptions
$$|\Cc_{n}(2;n^{1/2-\delta})|=o(\Nc_{n}),\; |\Cc_{n}(4;n^{1/2-\delta})|=o(\Nc_{n}^{2}), \;\text{and}\; |\Cc_{n}(6;n^{1/2-\delta})|=o(\Nc_{n}^{4}).$$ In fact, using a more combinatorial approach ~\cite{F-P-S-S-Z}, it is possible to prove
that if for some choice of $0<2\eta<\delta\le 1$ we have that both
$\Cc_{n}(2;n^{1/2-\delta/2+\eta}) =  \Cc_{n}(3;n^{1/2-\delta+\eta})= \varnothing$ are empty, then
$$|\Cc_{n}(6;n^{1/2-\delta})| \ll \Nc_{n}^{11/3} = o(\Nc_{n}^{4}),$$ which, on one hand, is a generic condition on $n$,
and on the other hand, in light of
the aforementioned result ~\cite[Lemma $5.2$]{R-W2017}, is sufficient for the conclusions of Theorem \ref{thm:toral univ spec corr} to hold.

\vspace{2mm}

It seems likely that by combining the ideas of the proof of Theorem \ref{thm:quasi-corr small} with some ideas in ~\cite{Gr-Wi}
it would be possible to shrink the balls faster: prove Theorem ~\ref{thm:toral univ gen}
for $s>\frac{\log{n}^{A}}{n^{1/2}}$ for generic $n$ for some $A\gg 0$ sufficiently big, i.e. save a power of $\log{n}$ rather than of $n$.
We believe that, in light of the results (and the techniques) presented in ~\cite{Gr-Wi} it is conceivable (if not likely)
that there is a phase transition: a number $A_{0}>0$ such that the conclusions of Theorem \ref{thm:toral univ gen} hold for generic $n$
uniformly for all $s>\frac{(\log{n})^{A}}{\sqrt{n}}$ with $A>A_{0}$, and fail for generic $n$ for $s=\frac{(\log{n})^{A}}{\sqrt{n}}$,
$A<A_{0}$. We leave all these questions to be addressed elsewhere. Finally, we note that all the methods presented in this manuscript work (resp. uniformly) unimpaired for generic smooth (shrinking) domains (as a replacement for discs), as long as the Fourier transform of the characteristic function is (resp. uniformly) decaying at infinity.

\vspace{2mm}

\qnew{

We believe that Theorem \ref{thm:quasi-corr small}
is of considerable independent interest. Other than the results contained in this paper, Theorem \ref{thm:quasi-corr small}
could be used in order to establish small-scale analogues of various other recently established and forthcoming results, though slightly restricted
in terms of energy levels (density $1$ sequence $S'\subseteq S$ rather than the whole of $S$).
As a concrete application, Planck-scale analogue of Bourgain's de-randomisation technique ~\cite{Bo,B-W} could be used for counting the number of nodal domains for (deterministic) ``flat" toral eigenfunctions ~\cite{B-B-W}. Further, it also implies \cite{B-B-W} that the value distribution of
the restriction of such flat eigenfunctions to a smooth curve with nowhere vanishing curvature is asymptotically Gaussian (convergence in the
sense of ~\cite{Bo,B-W}), which, in turn, as we hope, has other powerful applications.

}

\subsection{Outline of the proofs of the main results}

The proof of Theorem \ref{thm:toral univ spec corr} consists of two main steps. In the first step
(\S\ref{sec:proofs prelim}-\S\ref{sec:variance proofs}) we employ the Kac-Rice formula in order to express
the variance of $\Lc_{n;s}$ in terms of an integral on $B(s)\times B(s)$ of the $2$-point correlation function and study its asymptotic behaviour to yield \eqref{eq:Var asympt gen}; this step is analogous to ~\cite{K-K-W} posing new challenges for integrating the $2$-point correlation function on a restricted domain. In the second step (\S\ref{sec:full corr proof}) the full correlation \eqref{eq:restr nod len full corr full}
result is established.

\vspace{2mm}

A significant part of the first step is done in ~\cite{K-K-W}: it yields a point-wise asymptotic expansion
\eqref{eq:K2=1/4+L2+eps} for the $2$-point correlation function, provided that $|r_{n}(x)|$ is bounded away from $1$;
this eventually reduces the question of the asymptotic behaviour of $\var(\Lc_{n;s})$ to evaluating some moments of $r_{n}$
and its various derivatives, restricted to $B(s)\times B(s)$, provided that we avoid the ``singular set", i.e. $(x,y)$
such that $|r_{n}(x-y)|$ is arbitrarily close to $1$. As it was mentioned in \S\ref{sec:rest mom quasi corr}
evaluating the restricted moments
is a significant challenge of number theoretic nature; here we use the full strength of the assumptions of Theorem \ref{thm:toral univ spec corr}
on $n$.

To bound the contribution of the singular set we modify the approach
in ~\cite{O-R-W}, partitioning the singular set into small cubes of side length commensurable with $\frac{1}{\sqrt{n}}$.
One challenge here is that $B(s)\times B(s)$ could not be tiled by cubes; we resolve this by tiling a slightly excised
set, not beyond $B(2s)\times B(2s)$, using the latter in order to bound the total measure of the singular set.
We also simplify and improve our treatment of the singular set as compared to ~\cite{K-K-W},
following some ideas from ~\cite{R-W2014}: we use the Lipschitz continuity property satisfied by $r_{n}$ to bound
the total measure of the singular set, and also apply the partition into cubes on the singular set only
(as opposed to the full domain of integration).

\vspace{2mm}

After the variance of $\Lc_{n;s}$ has successfully been analysed \eqref{eq:Var asympt gen}, there are two ways to further proceed to establishing the full correlation result \eqref{eq:restr nod len full corr full}.
On one hand we may follow along the steps of ~\cite{M-P-R-W} to evaluate the Wiener chaos expansion of $\Lc_{n;s}$;
performing this we find that, under the assumptions of Theorem \ref{thm:toral univ spec corr}, the main terms of
its projection $\Lc_{n;s}[4]$ onto the $4$th Wiener chaos, dominating the fluctuations of $\Lc_{n;s}$,
recover, up to a scaling factor, the projection $\Lc_{n}[4]$ of the total nodal length of $f_{n}$ to the
$4$th Wiener chaos. This, in particular, implies the full correlation result \eqref{eq:restr nod len full corr full}.

On the other hand, now that much computational work has already been done, we might reuse
the precise information on the $2$-point correlation function ~\cite{K-K-W} to simplify
the proofs drastically by directly evaluating the correlation between $\Lc_{n;s}$ and $\Lc_{n}$ without decomposing them
into their respective Wiener chaos components. Equivalently, we evaluate the covariance $\cov(\Lc_{n;s},\Lc_{n})$
by employing the (suitably adapted) Kac-Rice formula once again; using the group structure of the torus, this approach
yields the intriguing identity
\begin{equation*}
\cov(\Lc_{n;s},\Lc_{n}) = (\pi s^{2}) \cdot \var(\Lc_{n}),
\end{equation*}
which, together with \eqref{eq:Var asympt gen} and \eqref{eq:var(tot nod length)}, recovers \eqref{eq:restr nod len full corr full}.

\section{Proof of Theorem \ref{thm:toral univ spec corr}}

\subsection{Preliminaries}
\label{sec:proofs prelim}

Recall that the covariance function $r=r_{n}$ of $f_{n}$ is given by \eqref{eq:rn covar def}. We further define the gradient
\begin{equation}
\label{eq:D grad def, sum lat}
D=D_{n;1\times 2}(x) = \nabla r_{n}(x) = \frac{2\pi i}{\Nc_{n}} \sum\limits_{\|\lambda\|^{2}=n} e(\langle \lambda,x\rangle) \cdot \lambda,
\end{equation}
and the Hessian
\begin{equation*}
H=H_{n;2\times 2}(x) = \left( \frac{\partial^{2}r_{n}}{\partial x_{i}\partial x_{j}} \right) = -\frac{4\pi^{2}}{\Nc_{n}}
 \sum\limits_{\|\lambda\|^{2}=n} e(\langle \lambda,x\rangle) (\lambda^{t}\lambda),
\end{equation*}
and the $2\times 2$ blocks (all depending on $n$, and evaluated at $x\in \Tb^{2}$)
\begin{equation}
\label{eq:XY block def}
X = - \frac{2}{E_{n}(1-r_{n}^{2})} D^{t}D, \; Y = -\frac{2}{E_{n}}\left(H+\frac{r_{n}}{1-r_{n}^{2}}D^{t}D\right).
\end{equation}
Finally, let $\Omega$ be the matrix
\begin{equation}
\label{eq:Omega covar V1,V2}
\Omega=\Omega_{n;4\times 4}(x) = I + \left(\begin{matrix}
X &Y \\Y &X \end{matrix}\right);
\end{equation}
it is ~\cite[Equalities (24), (25)]{K-K-W} the normalised covariance matrix of $(\nabla f_{n}(0),\nabla f_{n}(x))$ conditioned on $f_{n}(0)=f_{n}(x)=0$.

\vspace{2mm}

\new{In the following lemma we evaluate the variance of the restricted length $\Lc_{n;s}$; it is analogous to
\cite[Proposition $3.1$]{K-K-W} and \cite[Proposition $5.2$]{R-W2008} which give the variance of the total length
$\Lc_{n;s}$ except that, accordingly, the domain of integration in \eqref{eq:var(nod length)=int(K2)}
is restricted to $B(s)\times B(s)$, rather than the full $\Tb\times \Tb$ (reducing to $\Tb$ by stationarity).
The proof of these works is sufficiently robust to cover our case unimpaired, and thereupon conveniently omitted in this manuscript.}

\begin{lemma}
\label{lem:nodal var 2pnt corr norm}
For every $s>0$ we have
\begin{equation}
\label{eq:var(nod length)=int(K2)}
\var(\Lc_{n;s}) = \frac{E_{n}}{2}\int\limits_{B(s)\times B(s)} \left(K_{2}(x-y)-\frac{1}{4}\right)dxdy,
\end{equation}
where the (normalised) $2$-point correlation function is
\begin{equation}
\label{eq:K2 2pnt corr norm def}
K_{2}(x)=K_{2;n}(x) = \frac{1}{2\pi \sqrt{1-r_{n}(x)^{2}}} \cdot \E[\|V_{1}\|\cdot \|V_{2}\|],
\end{equation}
where $(V_{1},V_{2})\in \R^{2}\times \R^{2}$ is a centred $4$-variate Gaussian vector, whose covariance is given by
\eqref{eq:Omega covar V1,V2}, with $X$ and $Y$ given by \eqref{eq:XY block def}.
\end{lemma}

\subsection{Singular set}
\label{sec:singular set}

Recall that the nodal length variance (restricted to shrinking balls) is given by \eqref{eq:var(nod length)=int(K2)},
where the (normalised) $2$-point correlation function $K_{2}(x)$ is given by \eqref{eq:K2 2pnt corr norm def}, and
$(V_{1},V_{2})\in \R^{2}\times \R^{2}$ is centred Gaussian with covariance \eqref{eq:Omega covar V1,V2}, with $X$ and $Y$ given by
\eqref{eq:XY block def}. It is possible to expand $$\E[\|V_{1}\|\cdot \|V_{2}\|]$$ into a degree-$4$ Taylor polynomial as a
function of the (small) entries of
$X$ and $Y$ (these are the various derivatives of $r_{n}$), and, provided that the absolute value $|r_{n}(x)|$
is bounded away from $1$, we may write
\begin{equation*}
\frac{1}{\sqrt{1-r}} = 1+\frac{1}{2}r^{2}+\frac{3}{8}r^{4}+O(r^{6}).
\end{equation*}
These two combined yield a point-wise approximate \eqref{eq:K2=1/4+L2+eps} of $K_{2}(x)$,
provided that $|r_{n}(x)|$ is bounded away from $1$,
a condition that is satisfied for ``most" of $(x,y)\in B(s)\times B(s)$ (``nonsingular set" $(B(s)\times B(s))\setminus B_{sing}$,
see Lemma \ref{lem:sing set bound prop} below),
and we may integrate the Taylor polynomial of $x-y$ over the nonsingular set
to yield an approximation for the integral on the r.h.s. of \eqref{eq:var(nod length)=int(K2)}
while bounding the contribution of the singular set.

For the singular set $B_{sing}$ (where $|r_{n}(x)|$ is close to $1$) we only have an easy bound \eqref{eq:K2<<1/sqrt(1-r^2)},
and merely bounding its measure is insufficient for bounding its contribution
to the integral on the r.h.s. of \eqref{eq:var(nod length)=int(K2)}. We resolve this obstacle by observing that if $|r_{n}(x_{0})|$
is close to $1$, then it is so on the whole ($4d$) cube around $x_{0}$ of size length commensurable to $\frac{1}{\sqrt{n}}$, by
the Lipschitz property of $r_{n}$. This allows us to partition $B_{sing}$ into ``singular" cubes of side length commensurable with $\frac{1}{\sqrt{n}}$, possibly excising $B(s)\times B(s)$, though not beyond $B(2s)\times B(2s)$. We might then bound the number of
singular cubes using a simple Chebyshev's inequality bound \eqref{eq:meas(sing)<<R6 mom} via an appropriate $6$th moment
(it is $\Rc_{n}(6;2s)$ as $B_{sing}\subseteq B(2s)\times B(2s)$) while controlling
the contribution of a single singular cube to the integral on the r.h.s. of \eqref{eq:var(nod length)=int(K2)}
(as opposed to a point-wise bound).

The presented analysis is simplified and improved compared to ~\cite{R-W2008,K-K-W} in the following ways,
borrowing in particular some ideas from ~\cite{R-W2014}. First, only the singular set is partitioned into cubes as opposed to
the whole domain of integration (e.g. $B(s)\times B(s)$ in our case), since the point-wise estimate \eqref{eq:K2=1/4+L2+eps}
might be integrated on the nonsingular set to yield a precise estimate for its contribution to the integral on the r.h.s.
of \eqref{eq:var(nod length)=int(K2)}. The Lipschitz property of $r_{n}$ simplifies the partition argument of the singular set
into cubes, with no need to bound the individual cosines in \eqref{eq:rn covar def}.
Finally, working with (shrinking) subdomains $B(s)$ of the torus poses a problem while tiling the said domain ($B(s)\times B(s)$) into cubes;
since $s>n^{-1/2+\epsilon}$ we might still partition $B(s)\times B(s)$ into cubes of side length commensurable to $\frac{1}{\sqrt{n}}$ without
excising the domain of integration beyond $B(2s)\times B(2s)$. We start from the definition of the singular set.

\begin{definition}[Singular set]

Let $s>n^{-1/2+\epsilon}$ and choose
\begin{equation}
\label{eq:F cube side length}
F=F(n)=\frac{1}{c_{0}}\cdot \sqrt{n}
\end{equation}
a large integer, with $c_{0}>0$ a sufficiently small constant (that will be fixed throughout the rest of this manuscript).

\begin{enumerate}

\item A point $(x,y)\in B(s)\times B(s)$ is singular if $|r_{n}(x-y)| > \frac{7}{8}$ (say).

\item Let $$B(s)\times B(s) \subseteq \bigcup\limits_{i\in I}\Bc_{i} \subseteq B(2s)\times B(2s)$$
be a covering of $B(s)\times B(s)$ by ($4$d) cubes $\{ \Bc_{i}\}$ of side length $\frac{1}{F}$. We say that a cube $\Bc_{i}$ is singular
if it contains a singular point $x\in \Bc_{i}$.

\item Let $I'\subseteq I$ be the collection of all indices $i\in I$ such that $\Bc_{i}$ is singular. We define
the singular set
\begin{equation*}
B_{sing}(s)=B_{n;sing}(s) =\bigcup\limits_{i\in I'} \Bc_{i} \subseteq B(2s)\times B(2s).
\end{equation*}
to be the union of all singular cubes.

\end{enumerate}

\end{definition}

\begin{lemma}
\label{lem:sing set bound prop}

Let $F$ be as above \eqref{eq:F cube side length}, with $c_{0}$ sufficiently small.

\begin{enumerate}

\item
\label{it:sing Lipsh all points}
If $\Bc_{i} \subseteq B_{sing}$ is singular then for all $(x,y)\in \Bc_{i}$ we have $|r_{n}(x-y)| > \frac{1}{2}$.

\item The measure of the singular set is bounded by
\begin{equation}
\label{eq:meas(sing)<<R6 mom}
\meas(B_{sing}) \ll_{c_{0}} \Rc_{n}(6;2s)
\end{equation}
the $6$th moment \eqref{eq:Rc moments rn def} of $r_{n}$ on $B(2s)$.

\item The number $|I'|$ of singular cubes is bounded by
\begin{equation}
\label{eq:sing cubes << 6th mom}
|I'| \ll_{c_{0}} F^{4}\cdot \Rc_{n}(6;2s).
\end{equation}

\end{enumerate}

\end{lemma}

\begin{proof}

For \eqref{it:sing Lipsh all points} we note that $r_{n}$ is Lipschitz in all the variables with constant $\ll \sqrt{n}$.
Hence if $$|r_{n}(x_{0}-y_{0})|> \frac{3}{4}$$ for some $(x_{0},y_{0}) \in \Bc_{i}$, then for all $(x,y) \in \Bc_{i}$ we
have for some $C,C'$ absolute constants:
$$|r_{n}(x-y)|> \frac{3}{4} - C\sqrt{n} \|(x,y)-(x_{0},y_{0})\| \ge  \frac{3}{4} - C'\sqrt{n}\cdot \frac{c_{0}}{\sqrt{n}} > \frac{1}{2},$$
provided that $c_{0}$ was chosen sufficiently small.
The estimate \eqref{eq:meas(sing)<<R6 mom} now follows from the above and Chebyshev's inequality; note that the $6$th moment
$\Rc_{n}(6;2s)$ is over $B(2s)\times B(2s)$ rather than $B(s)\times B(s)$, in light of the fact that $B_{sing}$ might not be contained
in $B(s)\times B(s)$. Finally, \eqref{eq:sing cubes << 6th mom} follows from \eqref{eq:meas(sing)<<R6 mom} bearing in mind that
the measure of each singular cube in $B_{sing}$ is of order $\frac{1}{F^{4}}$.

\end{proof}

\subsection{Moments of $r$ and its derivatives along the shrinking balls $B(s)$}

Our final ingredient for the proof of the variance part of Theorem \ref{thm:toral univ spec corr} is evaluating certain moments of
$r_{n}$, $X_{n}$ and $Y_{n}$ in $B(s)$ in Lemma \ref{lem:r der moments}. The proof of Lemma \ref{lem:r der moments}
will be given in Appendix \ref{apx:r der moments}.

\begin{lemma}[Cf. ~\cite{K-K-W}, Lemma $4.6$, Lemma $5.4$]
\label{lem:r der moments}

Let $\epsilon>0$,
\begin{equation}
\label{eq:delt<eps}
\delta< \epsilon,
\end{equation}
and $S'\subseteq S$ a sequence of energy levels such that
for all $n\in S'$ the hypotheses $\Ac(n;2,n^{1/2-\delta/2})$ and $\Ac(n;6,n^{1/2-\delta})$ in Definition \ref{def:separatedness Ac} are satisfied. Then,
for all $A>0$ and uniformly for all $s>n^{-1/2+\epsilon}$, along
$n\in S'$ the following moments of $r_{n}$, $X_{n}$, $Y_{n}$ observe the following asymptotics, with constants involved in the `$O$'-notation
depending only on $A,\epsilon,\delta$:

\begin{enumerate}[$1.$]

\item
\begin{equation}
\label{eq:rn 2nd mom Bs}
\int\limits_{B(s)\times B(s)} r_{n}(x-y)^{2}dxdy = (\pi s^{2})^2 \cdot \frac{1}{\Nc_{n}}\left( 1 + O\left(\frac{1}{\Nc_{n}^{A}} \right) \right),
\end{equation}

\item
\begin{equation}
\begin{split}
\label{eq:rn 4th mom Bs}
\int\limits_{B(s)\times B(s)} r_{n}(x-y)^{4}dxdy &= (\pi s^2)^{2} \cdot \frac{|\Dc_{n}(4)|}{\Nc_{n}^{4}}\left(1 +
O\left(\frac{1}{\Nc_{n}^{A}} \right)\right) \\&= (\pi s^2)^{2} \cdot \frac{3}{\Nc_{n}^{2}}\left(1 +
O\left(\frac{1}{\Nc_{n}} \right)\right).
\end{split}
\end{equation}

\item
\begin{equation}
\label{eq:rn 6th mom Bs}
\Rc_{n}(6;s)=\int\limits_{B(s)\times B(s)} r_{n}(x-y)^{6}dxdy = O\left(s^4\cdot \frac{1}{\Nc_{n}^{5/2}}\right),
\end{equation}

\item
\begin{equation}
\label{eq:mom tr(X)}
\int\limits_{B(s)\times B(s)}\tr{X_{n}(x-y)}dxdy =(\pi s^{2})^{2}\left(-\frac{2}{\Nc_{n}}-\frac{2}{\Nc_{n}^{2}} +
O\left(  \frac{1}{\Nc_{n}^{5/2}} \right)   \right),
\end{equation}

\item
\begin{equation}
\label{eq:mom tr(Y^2)}
\int\limits_{B(s)\times B(s)}\tr(Y_{n}(x-y)^{2})dxdy =(\pi s^{2})^{2}\left(\frac{4}{\Nc_{n}}-\frac{4}{\Nc_{n}^{2}} +
O\left(  \frac{1}{\Nc_{n}^{5/2}} \right)   \right),
\end{equation}

\item
\begin{equation}
\label{eq:mom tr(XY^2)}
\int\limits_{B(s)\times B(s)}\tr(X_{n}(x-y)Y_{n}(x-y)^{2})dxdy =(\pi s^{2})^{2}\left(-\frac{4}{\Nc_{n}^{2}} +
O\left(  \frac{1}{\Nc_{n}^{5/2}} \right)   \right),
\end{equation}

\item
\begin{equation}
\label{eq:mom tr(X^2)}
\int\limits_{B(s)\times B(s)}\tr(X_{n}(x-y)^{2})dxdy =(\pi s^{2})^{2}\left(\frac{8}{\Nc_{n}^{2}} +
O\left(  \frac{1}{\Nc_{n}^{5/2}} \right)   \right),
\end{equation}

\item
\begin{equation}
\label{eq:mom tr(Y^4)}
\int\limits_{B(s)\times B(s)}\tr(Y_{n}(x-y)^{4})dxdy =(\pi s^{2})^{2}\left(\frac{2(11+\widehat{\new{\tau_{n}}}(4)^{2})}{\Nc_{n}^{2}} +
O\left(  \frac{1}{\Nc_{n}^{5/2}} \right)   \right),
\end{equation}

\item
\begin{equation}
\label{eq:mom tr(Y^2)^2}
\int\limits_{B(s)\times B(s)}\tr(Y_{n}(x-y)^{2})^{2}dxdy =(\pi s^{2})^{2}\left(  \frac{4(7+\widehat{\new{\tau_{n}}}(4)^{2})}{\Nc_{n}^{2}}
+ O\left(  \frac{1}{\Nc_{n}^{5/2}} \right) \right),
\end{equation}

\item
\begin{equation}
\label{eq:mom tr(X)tr(Y^2)}
\int\limits_{B(s)\times B(s)}\tr(X_{n}(x-y))\tr(Y_{n}(x-y)^{2})dxdy =(\pi s^{2})^{2}\left(  -\frac{8}{\Nc_{n}^{2}}
+ O\left(  \frac{1}{\Nc_{n}^{5/2}} \right) \right),
\end{equation}

\item
\begin{equation}
\label{eq:mom r^2 tr(X)}
\int\limits_{B(s)\times B(s)}r_{n}(x-y)^{2}\tr{X_{n}(x-y)}dxdy =(\pi s^{2})^{2}\left(  -\frac{2}{\Nc_{n}^{2}}
+ O\left(  \frac{1}{\Nc_{n}^{5/2}} \right) \right),
\end{equation}

\item
\begin{equation}
\label{eq:mom r^2 tr(Y^2)}
\int\limits_{B(s)\times B(s)}r_{n}(x-y)^{2}\tr(Y_{n}(x-y)^{2})dxdy =(\pi s^{2})^{2}\left(  \frac{8}{\Nc_{n}^{2}}
+ O\left(  \frac{1}{\Nc_{n}^{5/2}} \right) \right),
\end{equation}

\item
\begin{equation}
\label{eq:mom tr(X^3)}
\int\limits_{B(s)\times B(s)}\tr{X_{n}(x-y)^3}dxdy =(\pi s^{2})^{2} \cdot O\left(  \frac{1}{\Nc_{n}^{5/2}} \right),
\end{equation}

\item
\begin{equation}
\label{eq:mom tr(Y^6)}
\int\limits_{B(s)\times B(s)}\tr{Y_{n}(x-y)^6}dxdy =(\pi s^{2})^{2} \cdot O\left(  \frac{1}{\Nc_{n}^{5/2}} \right).
\end{equation}

\end{enumerate}

\end{lemma}

\subsection{Proof of the variance part \eqref{eq:Var asympt gen} of Theorem \ref{thm:toral univ spec corr}}
\label{sec:variance proofs}

\begin{lemma}[~\cite{K-K-W}, Lemma $3.2$]
\label{lem:K2=O(1)}
The matrices $X_{n}$ and $Y_{n}$ are uniformly bounded (entry-wise), i.e.
\begin{equation}
\label{eq:X,Y=O(1)}
X_{n}(x),Y_{n}(x) = O(1),
\end{equation}
for all $x\in \Tb^{2}$, where the constant involved in the `$O$'-notation is absolute. In particular
\begin{equation}
\label{eq:K2<<1/sqrt(1-r^2)}
K_{2}(x) \ll \frac{1}{\sqrt{1-r_{n}(x)^{2}}}.
\end{equation}
\end{lemma}

\begin{lemma}[Cf. Lemma \ref{lem:K2=O(1)} and \new{\cite[\S6.4]{O-R-W}}]
\label{lem:contr 2pnt corr sing cube}
Let $K_{2}$ be the $2$-point correlation function \eqref{eq:K2 2pnt corr norm def}, and $\Bc_{i}$, $i\in I'$ a singular cube. Then
\begin{equation*}
\int\limits_{\Bc_{i}}K_{2}(x-y)dxdy \ll_{c_{0}} \frac{1}{F^{3}\sqrt{n}}.
\end{equation*}
\end{lemma}

For brevity of notation in what follows we will sometimes suppress the dependency of various variables on $x$ or $n$,
e.g. $r$ or $r(x)$ will stand for $r_{n}(x)$, and $X$ will denote the $2\times 2$ matrix $X_{n}(x)$ in \eqref{eq:XY block def}.

\begin{proposition}[Cf. ~\cite{K-K-W}, Proposition $4.5$]
\label{prop:Kac-Rice K2 asymp}
For every $x\in\Tb^{2}$ such that $|r_{n}(x)| < \frac{1}{4}$, the $2$-point correlation function $K_{2}(x)=K_{2;n}(x)$ satisfies the asymptotic expansion
\begin{equation}
\label{eq:K2=1/4+L2+eps}
K_{2}(x) =\frac{1}{4}+ L_{2}(x)+\epsilon(x),
\end{equation}
where
\begin{equation}
\label{eq:L2 approx def}
\begin{split}
L_{2}(x) &= \frac{1}{8}\bigg( r^{2}+ \tr{X}+ \frac{\tr(Y^{2})}{4}+\frac{3}{4}r^{4}-\frac{\tr(XY^{2})}{8} -\frac{\tr(X^{2})}{16}+\frac{\tr(Y^{4})}{128}\\&+\frac{\tr(Y^{2})^{2}}{256} - \frac{\tr(X)\cdot \tr(Y^{2})}{16}+\frac{1}{2}r^{2}\tr(X)+\frac{1}{8}r^{2}\tr(Y^{2})  \bigg)
\end{split}
\end{equation}
and
\begin{equation}
\label{eq:eps error def}
|\epsilon(x)| = O(r^{6}+\tr(X^{3})+\tr(Y^{6})).
\end{equation}
\end{proposition}

\begin{proof}[Proof of the variance part \eqref{eq:Var asympt gen} of Theorem \ref{thm:toral univ spec corr}]

We invoke Lemma \ref{lem:nodal var 2pnt corr norm} and separate the singular contribution to write
\begin{equation}
\label{eq:var 2pnt corr sep sing}
\begin{split}
\var(\Lc_{n;s}) &= \frac{E_{n}}{2}\int\limits_{B(s)\times B(s)} \left(K_{2}(x-y)-\frac{1}{4}\right)dxdy
\\&=\frac{E_{n}}{2}\int\limits_{B_{sing}\cap (B(s)\times B(s))} \left(K_{2}(x-y)-\frac{1}{4}\right)dxdy
\\&+ \frac{E_{n}}{2}\int\limits_{(B(s)\times B(s))\setminus B_{sing}} \left(K_{2}(x-y)-\frac{1}{4}\right)dxdy.
\end{split}
\end{equation}
Now we bound the contribution of singular set (former integral on the r.h.s. of \eqref{eq:var 2pnt corr sep sing}) by
\begin{equation}
\label{eq:int K2 sing << R6 mom}
\begin{split}
&\int\limits_{B_{sing}\cap (B(s)\times B(s))} \left(K_{2}(x-y)-\frac{1}{4}\right)dxdy \\&=
\int\limits_{B_{sing}\cap (B(s)\times B(s))} K_{2}(x-y)dxdy + O(\meas(B_{sing}\cap (B(s)\times B(s))))
\\&\le \int\limits_{B_{sing}}K_{2}(x-y)dxdy + O(\meas(B_{sing})) \\&\ll
F^{4}\cdot \Rc_{n}(6;s) \cdot \frac{1}{F^{3}\sqrt{n}} + \Rc_{n}(6;2s) \ll \Rc_{n}(6;2s),
\end{split}
\end{equation}
where we employed \eqref{eq:meas(sing)<<R6 mom} and \eqref{eq:sing cubes << 6th mom} of Lemma
\ref{lem:sing set bound prop}, Lemma \ref{lem:contr 2pnt corr sing cube}, and \eqref{eq:F cube side length}.

On the nonsingular range $(B(s)\times B(s))\setminus B_{sing}$ (latter integral on the r.h.s. of \eqref{eq:var 2pnt corr sep sing})
we have $|r_{n}(x-y)| < \frac{3}{4}$, hence we are eligible to invoke Proposition \ref{prop:Kac-Rice K2 asymp} to write
\begin{equation}
\label{eq:int K2-1/2 L,eps nonsing}
\begin{split}
&\int\limits_{(B(s)\times B(s))\setminus B_{sing}} \left(K_{2}(x-y)-\frac{1}{4}\right)dxdy \\&=
\int\limits_{(B(s)\times B(s))\setminus B_{sing}}\left(L_{2}(x-y)+\epsilon(x-y)\right)dxdy \\&=
\int\limits_{B(s)\times B(s)}\left(L_{2}(x-y)+\epsilon(x-y)\right)dxdy + O(\meas(B_{sing}))
\\&= \int\limits_{B(s)\times B(s)}\left(L_{2}(x-y)+\epsilon(x-y)\right)dxdy + O(\Rc_{n}(6;2s)),
\end{split}
\end{equation}
by \eqref{eq:X,Y=O(1)}, and \eqref{eq:meas(sing)<<R6 mom}. Consolidating the contributions \eqref{eq:int K2 sing << R6 mom} and
\eqref{eq:int K2-1/2 L,eps nonsing} of the singular and nonsingular ranges respectively to the integral
in \eqref{eq:var 2pnt corr sep sing} we obtain
\begin{equation}
\label{eq:int K2-1/4 by L2, eps, err}
\begin{split}
&\int\limits_{B(s)\times B(s)} \left(K_{2}(x-y)-\frac{1}{4}\right)dxdy
\\&= \int\limits_{B(s)\times B(s)}\left(L_{2}(x-y)+\epsilon(x-y)\right)dxdy + O(\Rc_{n}(6;2s)).
\end{split}
\end{equation}

By the very definition \eqref{eq:L2 approx def} of $L_{2}$ and \eqref{eq:eps error def}, the above estimate \eqref{eq:int K2-1/4 by L2, eps, err}
relates the nodal length variance to evaluating moments of the encountered expressions along shrinking balls. The latter is precisely the statement of Lemma \ref{lem:r der moments}, under the hypotheses $\Ac(n;2,n^{1/2-\delta/2})$ and $\Ac(n;6,n^{1/2-\delta})$ of Theorem
\ref{thm:toral univ spec corr}, that so far haven't been exploited. Lemma \ref{lem:r der moments} is instrumental for
the precise asymptotic evaluation of the integral $$\int\limits_{B(s)\times B(s)}L_{2}(x-y)dxdy$$ and
bounding the contribution $$\int\limits_{B(s)\times B(s)}\epsilon(x-y)dxdy$$ and $\Rc_{n}(6;2s)$ of the error terms in the following way.
First,
\begin{equation}
\label{eq:eps L1 mom small}
\int\limits_{B(s)\times B(s)}\epsilon(x-y)dxdy \ll
\int\limits_{B(s)\times B(s)}\left(r^{6}+\tr(X^{3})+\tr(Y^{6})\right)dxdy \ll s^{4}\frac{1}{\Nc_{n}^{5/2}}
\end{equation}
and
\begin{equation}
\label{eq:Rn(6,2s) small}
\Rc_{n}(6;2s) \ll s^{4}\cdot \frac{1}{\Nc_{n}^{5/2}}
\end{equation}
by \eqref{eq:rn 6th mom Bs} (applied both on $s$ and $2s$), \eqref{eq:mom tr(X^3)} and \eqref{eq:mom tr(Y^6)}.
Next,
\begin{equation}
\label{eq:L2 mom ball estimate}
\begin{split}
&\int\limits_{B(s)\times B(s)}L_{2}(x-y)dxdy = \frac{1}{8}\cdot (\pi s^{2})^{2} \cdot \bigg( \frac{1}{\Nc_{n}}
+\left(-\frac{2}{\Nc_{n}}-\frac{2}{\Nc_{n}^{2}} \right) \\&+\frac{1}{4}\cdot \left( \frac{4}{\Nc_{n}}-\frac{4}{\Nc_{n}^{2}}  \right)
+\frac{3}{4}\cdot \frac{3}{\Nc_{n}^{2}}  +\frac{1}{8}\cdot \frac{4}{\Nc_{n}^{2}} - \frac{1}{16}\cdot \frac{8}{\Nc_{n}^{2}}
+\frac{1}{128} \cdot   \frac{2(11+\widehat{\new{\tau_{n}}}(4)^{2})}{\Nc_{n}^{2}}\\& +
\frac{1}{256}\cdot \frac{4(7+\widehat{\new{\tau_{n}}}(4)^{2})}{\Nc_{n}^{2}}
+\frac{1}{16}\cdot \frac{8}{\Nc_{n}^{2}} - \frac{1}{2}\cdot  \frac{2}{\Nc_{n}^{2}} +\frac{1}{8}\cdot \frac{8}{\Nc_{n}^{2}} +
O\left( \frac{1}{\Nc_{n}^{5/2}}  \right)\bigg)
\\&= (\pi s^{2})^{2} \cdot \left(\frac{1+\widehat{\new{\tau_{n}}}(4)^{2}}{256\cdot \Nc_{n}^{2}}  + O\left( \frac{1}{\Nc_{n}^{5/2}}  \right)\right)
\end{split}
\end{equation}
by \eqref{eq:rn 2nd mom Bs}, \eqref{eq:mom tr(X)}, \eqref{eq:mom tr(Y^2)}, \eqref{eq:rn 4th mom Bs}, \eqref{eq:mom tr(XY^2)},
\eqref{eq:mom tr(X^2)}, \eqref{eq:mom tr(Y^4)}, \eqref{eq:mom tr(Y^2)^2}, \eqref{eq:mom tr(X)tr(Y^2)},
\eqref{eq:mom r^2 tr(X)} and \eqref{eq:mom r^2 tr(Y^2)}, with the $\frac{1}{\Nc_{n}}$ term vanishing.
Substituting \eqref{eq:L2 mom ball estimate}, \eqref{eq:eps L1 mom small} and \eqref{eq:Rn(6,2s) small} into
\eqref{eq:int K2-1/4 by L2, eps, err}, and then finally into the first equality of \eqref{eq:var 2pnt corr sep sing} yields the variance statement
\eqref{eq:Var asympt gen} of Theorem \ref{thm:toral univ spec corr}.

\end{proof}

\subsection{Proof of the full correlation part \eqref{eq:restr nod len full corr full} of Theorem \ref{thm:toral univ spec corr}}

\label{sec:full corr proof}

\begin{proof}

Let $\epsilon>0$ be given; we are going to show that for every $s>0$ we have the precise identity
\begin{equation}
\label{eq:cov(Lns,Ln)=vol(B)*var(L)}
\cov(\Lc_{n;s},\Lc_{n}) = (\pi s^{2}) \cdot \var(\Lc_{n}).
\end{equation}
Once \eqref{eq:cov(Lns,Ln)=vol(B)*var(L)} has been established, \eqref{eq:restr nod len full corr full} follows at once for $n$ satisfying
\new{\eqref{eq:Var asympt gen}} (valid for a generic sequence $\{n\}\subseteq S$) and \eqref{eq:var(tot nod length)},
uniformly for $s>n^{-1/2+\epsilon}$.

To show \eqref{eq:cov(Lns,Ln)=vol(B)*var(L)} we recall that
$K_{2}(x)=K_{2;n}(x)$ as in \eqref{eq:K2 2pnt corr norm def} is the (normalised) $2$-point correlation function, and that we have that
the total nodal length variance $\var(\Lc_{n})$ is given by ~\cite{R-W2008,K-K-W}
\begin{equation}
\label{eq:var(L)=int(K2)TxT}
\var(\Lc_{n})= \frac{E_{n}}{2}\int\limits_{\Tb^{2}} \left(K_{2}(x)-\frac{1}{4}\right)dx.
\end{equation}
(cf. \eqref{eq:var(nod length)=int(K2)}).
For the covariance we have the analogous formula
\begin{equation}
\label{eq:cov(Ls,L)=int(K2)BsxT}
\cov(\Lc_{n;s},\Lc_{n})= \frac{E_{n}}{2}\int\limits_{B(s)\times \Tb^{2}} \left(K_{2}(x-y)-\frac{1}{4}\right)dx.
\end{equation}
Since for every $x$ fixed, wherever $y$ varies along the torus so does $x-y$, \eqref{eq:cov(Ls,L)=int(K2)BsxT} reads
\begin{equation*}
\cov(\Lc_{n;s},\Lc_{n})= \vol(B(s))\cdot \frac{E_{n}}{2}\int\limits_{\Tb^{2}} \left(K_{2}(x)-\frac{1}{4}\right)dx =
(\pi s^{2})\cdot \var(\Lc_{n}),
\end{equation*}
by \eqref{eq:var(L)=int(K2)TxT}. This concludes the proof of \eqref{eq:cov(Lns,Ln)=vol(B)*var(L)}, which, as mentioned
above, implies \eqref{eq:restr nod len full corr full}.

\end{proof}

\section{Proof of Theorem \ref{thm:quasi-corr small}: bound for quasi-correlations}

\label{sec:qcorr proof}

\qnew{

Our first goal is to state a quantitative version of Theorem \ref{thm:quasi-corr small} (in terms of the exceptional number of
$n\in S$ not obeying the properties claimed by Theorem \ref{thm:quasi-corr small}), also controlling the possible weak-$*$ partial
limits of the respective $\{\tau_{n}\}_{n\in S'}$, namely measures $\nu_{s}$ introduced immediately below.

\begin{definition}
\label{measdef}
For $s \in [0, \pi/4]$ let the (symmetric) probability measure $\nu_s$ be given by
\begin{equation*}
\int_{S^1} f \ d \nu_s :=  \frac{1}{8s} \sum_{ j=0}^{3} \int_{-s+k \pi/2}^{s+k \pi/2} f(e^{i \theta}) \ d \theta \qquad (f \in C(\Sc^1)).
\end{equation*}

\end{definition}

We will \qcnew{adopt} the following conventions:

\begin{notation}
Throughout this section we will use the notation $[N]=\left\{ 1, \cdots , N \right\}$ for any natural number $N$ and write $n \asymp x$ to mean $n \in [x,2x]$. The shorthand $\log_{2} n:= \log \log n$ will be in use and, as usual, $\Omega(n)=\sum_{p^e ||n} e$ denotes the number of prime divisors of $n$ counted with multiplicity. We will say the numbers $\theta_1,..., \theta_r \in \R$ are $\gamma$-separated when $|\theta_j-\theta_i| \geq \gamma$ for all $i \neq j$.
\end{notation}

Theorem \ref{maincorr} is the announced quantitative version of Theorem \ref{thm:quasi-corr small}.

}

\qnew{
\begin{theorem}\label{maincorr}
Given any $0< \delta \leq 1$, $l \geq 2$, the following two properties hold.\\
a) The exceptional set $\mathcal{R}_N(l; \delta):=\left\{N \leq n \leq 2N:  \mathcal{C}_n(l;n^{(1-\delta)/2}) \neq \emptyset \right\} $
has size at most
\begin{equation}\label{RN}
|\mathcal{R}_N(l; \delta) | \ll \kappa^{L  }  (2L)! \  N^{1- \rho_0 (\delta, l)} (\log N)^{L +1} ,
\end{equation}
where $\kappa >0$ is an absolute constant, $L:=2^{l}$ and $\rho_0 (\delta, l)= \delta/(2 \cdot4^{L }  (l +1)  )$.\\
b) For any $s \in [0, \pi/4]$ there exists a sequence of natural numbers  \qnew{ $\{n_{k}\}_{k\ge 1}\subseteq S$ so that
$$ \mathcal{C}_{n_{k}}(l;n_{k}^{(1-\delta)/2}) =\emptyset $$
and $\tau_{n_{k}} \Rightarrow  \nu_s$. }
\end{theorem}
}

 \qcnew{
\begin{proof}[Proof of Theorem \ref{thm:quasi-corr small} assuming Theorem \ref{maincorr}]

Let $\delta \in (0,1/2)$ and $l\ge 2$ be given and define
$$S'(l, \delta)= \left\{n \in S:  \mathcal{C}_n(l;n^{1/2-\delta}) = \emptyset \right\}.$$
Part 3 of Theorem \ref{thm:quasi-corr small} is true by definition while part 1 follows immediately from the power saving in \eqref{RN}. Part 2 is proven as in \cite[Section 7.2]{K-K-W}: Let $s \in [0, \pi/4]$ be arbitrary and consider any sequence $(n_k)_{k \geq1} \subseteq S'(l, \delta) $ for which $\tau_{n_k} \Rightarrow  \nu_s$. Seeing how $\widehat{\tau_{n_k} }(4) \rightarrow \widehat{\nu_s}(4)$, the result follows from the continuity of the map $s \mapsto \ \widehat{\nu_s}(4) $ (with boundary values $\widehat{ \nu_{0} }(4)=1$, $\widehat{ \nu_{\pi/4} }(4)=0$), or, alternatively,
the explicit evaluation of $$\widehat{\nu_s}(4) =\frac{\sin(4s)}{4s}.$$

\end{proof}

}

The rest of section \ref{sec:qcorr proof} is dedicated to proving Theorem \ref{maincorr}.

\qnew{

\subsection{Preliminary results} We begin with two simple estimates.

\begin{lemma}\label{sectorbound}
Let $R \geq1$ and suppose \qcnew{that} the angles $0< \theta_1 < \theta_2 \leq 2 \pi$ are $1/R$-separated. Then the lattice points contained in the sector $\Gamma(R;\theta_1, \theta_2):=\left\{r e^{i \theta} : r \leq R, \  \theta \in (\theta_1, \theta_2) \right\}$ number at most
$$\left| \Gamma(R;\theta_1, \theta_2) \cap \Z[i] \right| \leq \kappa_1 (\theta_2-\theta_1) R^2$$
for some absolute constant $\kappa_1 >0$.
\begin{proof}
The number of Gaussian integers inside $\Gamma(R;\theta_1, \theta_2)$  is bounded by the area of the covering region $$\tilde{\Gamma}:=\left\{ z_1+z_2 \in \C : z_1 \in \Gamma(R;\theta_1, \theta_2), |z_2| \leq \sqrt{2} \right\}.$$
The area of $\tilde{\Gamma}$ is at most $c_1 ((\theta_2-\theta_1) R^2+ R)$ for some absolute constant $c_1$ and hence the result follows.
\end{proof}
\end{lemma}

\begin{lemma}\label{exppowersums}
Given any $ \alpha \geq 0$ and $x \geq 1$ we have the estimate
$$\sum_{  \substack{ a \in \Z[i] \setminus{0} \\  |a|\leq x }   } \frac{1}{|a|^{2- \alpha}} \leq \kappa_2 \ x^{\alpha} \log (2+x),$$
for some absolute constant $\kappa_2 >0$.
\begin{proof}
The bound is immediate for $1 \leq x <2 $ ( so long as $\kappa_2 \geq 4$) so we may assume \qcnew{that} $x \geq 2$.
Covering $[1,x]$ with dyadic intervals $[D,2D]=[2^i, 2^{i+1}]$ ( $i=0,..., \lfloor \log x /\log2 \rfloor $), we get

\begin{equation*}
\sum_{  \substack{ a \in \Z[i]\setminus{0} \\  |a| \leq x }   } \frac{1}{|a|^{2- \alpha}} \leq \sum_{  \substack{  1 \leq D \leq x   \\ \text{dyadic}   } }   \sum_{  \substack{ a \in \Z[i] \\  |a|  \asymp D }   } \frac{1}{D^{2-\alpha}} \ll \sum_{  \substack{  1 \leq D \leq x   \\ \text{dyadic}   } }  D^{\alpha} \ll  x^{\alpha} \log x.
\end{equation*}

\end{proof}

\end{lemma}

To prove Theorem \ref{maincorr}b we will require upper and lower bounds for the number of Gaussian primes in narrow sectors. These estimates can be found in \cite[Lemma 1]{Erdos-Hall} and \cite[Theorem 2]{HL}.
\begin{theorem}[\cite{Erdos-Hall}, \cite{HL}]\label{narrowprime}
Let $R>R_0$ be a sufficiently large number. Assuming \qcnew{that}
$$0\leq \alpha < \beta \leq \pi/2, \qquad \beta - \alpha \geq R^{-0.762},$$
the number of Gaussian primes in the sector $\Gamma(R, \alpha, \beta)$ is bounded from below and above by
$$c \frac{R^2 (\beta - \alpha) }{\log R} \leq \sum_{\pi \in \Gamma(R, \alpha, \beta)} 1 \leq C \frac{R^2 (\beta - \alpha) }{\log R}$$
for some absolute constants $c,C>0$.
\end{theorem}
As an immediate consequence of Theorem \ref{narrowprime} we obtain a lower bound for the number of $k$-almost Gaussian primes in $\Gamma(R,0,\beta)$ (provided $k$ is not too large and $\beta$ not too small).

\begin{lemma}\label{almostprimes}
Let $R>R_0$ be a sufficiently large number and suppose $1 \leq k \leq  (\log R)^{1/4}$. Then for any choice of angles $0 \leq \alpha_i < \beta_i \leq \pi/2$ $(i=1,...,k)$ satisfying
$$\frac{1}{(\log R)^2} \leq \beta_i - \alpha_i \leq \pi/2,$$
one has the lower bound
\begin{equation}\label{kalmost}
\sum_{  \substack{ \pi_i \in  \Gamma(R,\alpha_i,\beta_i),  i \leq k   \\ |\pi_1\cdots \pi_k | \leq R    }  } 1 \geq \left( \frac{ c }{\qcnew{4}}\right)^k \frac{R^2}{ \qcnew{ (\log R)^{3k} }  },
\end{equation}
where $c>0$ is the constant from Theorem \ref{narrowprime}.
\begin{proof}
Writing $\Phi(R)=\exp(\sqrt{\log R}) $ and applying Theorem \ref{narrowprime} one finds that
\begin{align*}
\sum_{  \substack{ \pi_i \in  \Gamma(R,\alpha_i,,\beta_i),  i \leq k   \\ |\pi_1\cdots \pi_k | \leq R    }  } 1 & \geq  \sum_{  \substack{ \pi_i \in  \Gamma(R,\alpha_i,,\beta_i),  i \leq k -1  \\  \Phi(R) \leq |\pi_1|, ..., |\pi_{k-1} | \leq \Phi(R)^2    }  } \sum_{  \substack{  \pi_k \in  \Gamma(R,\alpha_k,\beta_k) \\  |\pi_k| \leq R/ |\pi_1 \cdots \pi_{k-1} |  }   } 1 \\
\geq &\sum_{  \substack{ \pi_i \in  \Gamma(R,\alpha_i,\beta_i),  i \leq k -1  \\  \Phi(R) \leq |\pi_1|, ..., |\pi_{k-1} | \leq \Phi(R)^2    }  } \frac{c R^2}{|\pi_1 \cdots \pi_{k-1} |^2  \qcnew{ (\log R)^3}   } \\
& \geq \frac{c R^2}{ \qcnew{ (\log R)^3 }\Phi(R)^{4(k-1)} }  \sum_{  \substack{ \pi_i \in  \Gamma(R,\alpha_i,\beta_i),  i \leq k -1  \\  \Phi(R) \leq |\pi_1|, ..., |\pi_{k-1} | \leq \Phi(R)^2    }  } 1.
\end{align*}
\qcnew{In the first inequality we used the restriction $k \leq (\log R)^{1/4}$ to ensure that $R/ |\pi_1 \cdots \pi_{k-1} | \geq R^{1-o(1)}$.} Invoking Theorem \ref{narrowprime} once again, we get
$$\sum_{  \substack{ \pi_i \in  \Gamma(R,\alpha_i,\beta_i),  i \leq k -1  \\  \Phi(R) \leq |\pi_1|, ..., |\pi_{k-1} | \leq \Phi(R)^2    }  } 1 \geq
\left( \frac{c \ \Phi(R)^4 }{2 \qcnew{  \log (\Phi(R)^2) (\log R)^2}  }\right)^{k-1},$$
which combined with the previous estimate yields \eqref{kalmost}.
\end{proof}

\end{lemma}

\subsection{An upper bound for $\mathcal{R}_N$}

Let $N$ be a large natural number and suppose $n \asymp N$. In order to show that $\mathcal{C}_n(l;n^{(1-\delta)/2})$
is generically empty we first recall that, typically, $n$ has $O( \log_{2} n)$ prime divisors (see \cite[Section III.3, Theorem 4]{Ten}). We will need the following quantitative result of Erd\H{o}s-S\'ark\"ozy to bound the number of $n \asymp N$ for which $\Omega(n)$ is unusually large.
\begin{theorem}\cite[Corollary 1]{ES}\label{norder}
For all $\qcnew{N} \geq 3, \ K \geq 1$ we have the estimate
$$\left| \left\{n \leq N :  \Omega(n)\geq K \right\} \right| \ll K^4 \frac{N \log N}{2^{K}}.$$
\end{theorem}

\subsubsection{Using the structure of $\mathcal{E}_n$}

We now turn to the study of $\Cc_n(l;n^{(1-\delta)/2})$ and decompose
$$n=\left| \prod_{j\leq k} \pi_j \right|^2$$
into a product of Gaussian primes. Rotating each prime $\pi_j$ over an integer multiple of $\pi/2$ we may assume \qcnew{that} $\pi_j$ lies in the first quadrant and makes an angle $0 \leq \theta_j < \pi/2$ with the x-axis.  Each lattice point $\xi \in \mathcal{E}_n$ takes the form
\begin{align}\label{Enpoints}
\xi=(i)^{\nu} n^{1/2} \prod_{j \leq k} e^{i \epsilon_j \theta_j }, \qquad \epsilon_j \in \left\{-1,1 \right\},\qquad j=1...,k, \qquad \nu=1,...,4.
\end{align}

As a result, each tuple $(\xi_1,...,\xi_{l}) \in \mathcal{E}_n^{l}$ corresponds to a choice of matrix  $\epsilon= (\epsilon^{(r)}_j )_{1\leq r \leq {l}, 1\leq j\leq k}$ with entries in $\left\{-1,1 \right\}$, and a choice of vector $\underline{\nu}=(\nu_{1 }, ..., \nu_{l }) \in [4]^{l}$. The representation in (\ref{Enpoints}) is not unique: different vectors $(\epsilon_j)_{j \leq k}$ can give rise to the same lattice point $\xi \in \mathcal{E}_n$. We refer the reader to \cite{Ci} for a precise description of $\mathcal{E}_n$ and note that Theorem \ref{maincorr} only delivers an upper bound for $\mathcal{R}_N$ which is why we allow for overcounting.

\vspace{2mm}

{\bf Rewriting the quasi-correlation condition.} Let $y \geq 1$. For a given tuple $(\xi_1,...,\xi_{l}) \in \mathcal{E}_n^{l}$ with associated matrix  $\epsilon= (\epsilon^{(r)}_j )_{1\leq r \leq {l}, 1\leq j\leq k}$ and vector $\underline{\nu} \in [4]^{l}$, we wish to express the inequalities \newline
$$0< \left\|\sum_{i=1}^{l} \xi_i \right\|/\sqrt{n} < y^{-\delta}$$ in terms of $\epsilon$ and $\underline{\nu}$. Introducing vectors $\underline{\eta}^{+}, \underline{\eta}^{-} \in \left\{-1,1 \right\}^{l}$ and $\underline{\nu}^{+},\underline{\nu}^{-} \in [4]^{l}$ we consider the more general conditions

\begin{equation}\label{conds}
\begin{aligned}
 &0< \left|\sum_{r \leq l } \eta_r^{+} \cos \left(\nu_r^{+} \pi /2+ \sum_{j \leq k} \epsilon^{(r)}_j \theta_j \right) \right| < y^{-\delta}, \qquad  \qquad (+) \\
&0< \left|\sum_{r \leq l } \eta_r^{-} \sin \left( \nu_r^{-} \pi /2+\sum_{j \leq k} \epsilon^{(r)}_j \theta_j \right) \right| < y^{-\delta} \ \   \qquad  \qquad (-)
\end{aligned}
\end{equation}
This slight generalisation will be useful in the proof of Proposition \ref{Spm} below.

Recalling the notation $\mathcal{R}_N(l; \delta):=\left\{n \asymp N:  \mathcal{C}_n(l;n^{(1-\delta)/2}) \neq \emptyset \right\} $, set $Y=N^{1/2}$ and let $K=K(N)$ be a large integer to be chosen later. It follows from Theorem \ref{norder} and the preceding discussion that
\begin{align}\label{eset}
|\mathcal{R}_N(l; \delta)| \ll \sum_{\substack{ \underline{\eta}^{+} \in \left\{-1,1 \right\}^{l}\\  \underline{\nu}^{+} \in [4]^{l}}      }   \     \sum_{k < K} \ \sum_{\epsilon= \epsilon^{(r)}_j  }
\sum_{  \substack{ |\pi_1|\leq ...\leq |\pi_k|  \\ |\pi_1 \cdots \pi_k| \asymp Y } }^{+} 1+ O\left(    K^4 \frac{N \log N}{2^{K}}   \right).
\end{align}
The superscript $+$ refers to the first condition in (\ref{conds}) (choosing $y=Y=N^{1/2}$) and the $(\pi_j)_{j \leq k}$ in the innermost sum range over Gaussian primes in the first quadrant.

\subsubsection{Quasi-correlations of Gaussian integers}

To prove Theorem \ref{maincorr}a it is enough (as will be shown in section \ref{complcorr} below) to bound the quantities
$$S_{\delta}^{\pm}(y; k, \epsilon, \underline{\eta}^{\pm}, \underline{\nu}^{\pm}):=  \sum_{ \substack{a_1,...,a_k \in \Z[i] \setminus \left\{0 \right\} \\  |\prod_j a_j | \asymp y }}^{\pm} 1,$$
where, as before, the superscripts $+$ and $-$ refer to the conditions in (\ref{conds}) and the angle $\theta_j$ belongs to the Gaussian integer $a_j$ (which is not assumed to be prime).

\begin{proposition}\label{Spm}
For any choice of parameters $y \geq 1, l \geq 2, k\leq 2^{l}$ and $\delta \in (0,1)$ we have the estimates
\begin{equation}\label{Sbound}
\max_{\epsilon, \underline{\eta}^{\pm}, \underline{\nu}^{\pm}} \left| S_{\delta}^{\pm}(y; k, \epsilon, \underline{\eta}^{\pm}, \underline{\nu}^{\pm}) \right| \leq \kappa_3^k \  (2k) !\  y^{2-\delta/4^k  } (\log(2+ 2y))^{k-1}
\end{equation}
where $\kappa_3 >0$ is an absolute constant.
\begin{proof}
Let $\kappa_1,\kappa_2 \geq 3$ be the constants appearing in Lemmas \ref{sectorbound} and \ref{exppowersums} and put $\kappa_3=4 \kappa_1 \kappa_2$. We will prove the estimate \eqref{Sbound} by induction on $k$, noting that the $k=1$ case is an immediate consequence of Lemma \ref{sectorbound}.\\
Assuming \eqref{Sbound} holds for all values up to $k-1$, let us verify the bound for $S_{\delta}^{+}(y; k, \epsilon, \underline{\eta}^{+}, \underline{\nu}^{+})$ (the treatment of $S_{\delta}^{-}$ is almost identical). A quick inspection of $S_{\delta}^{+}$ reveals that at least one of the summation variables, say $a_k$, must be large. More precisely, we may assume that
\begin{equation}\label{indfix}
|a_k| \geq y^{1/k} \geq y^{\delta/k}.
\end{equation}
Applying the addition formula for cosine we may write, for each $i \leq k$,
\begin{align}\label{separatek}
\sum_{r \leq l } \eta_r^{+} \cos \left(\nu_r^{+} \pi /2+ \sum_{j \leq k} \epsilon^{(r)}_j \theta_j \right) &=
\left[ \sum_{r \leq l } \eta_r^{+} \cos \left(\nu_r^{+} \pi /2+ \sum_{  \substack{ j \leq k  \\ j \neq i }} \epsilon^{(r)}_j \theta_j \right) \right] \cos(\theta_i) \\ \notag
&- \left[ \sum_{r \leq l } \eta_r^{+} \epsilon^{(r)}_k \sin \left(\nu_r^{+} \pi /2+ \sum_{\substack{ j \leq k \\ j \neq i  }} \epsilon^{(r)}_j \theta_j \right) \right] \sin(\theta_i) \\ \notag
&=: \hat{A}_{i} \cos(\theta_i) + \hat{B}_{i} \sin(\theta_i).
\end{align}
 Before proceeding with the argument we record the following observation.\\
{\bf Claim.} Either of the conditions $0<|\hat{A}_{i}|\leq y^{-\delta}$ or $0<|\hat{B}_{i}|\leq y^{-\delta}$ imply $|\prod_{j \neq i} a_j| \geq y^{\delta}$.\\
To prove the claim, we write $a_j= b_j+ i c_j$ so that $|\sin \theta_j|= |c_j|/ |a_j| $ and $|\cos \theta_j|= |b_j|/|a_j| $. Repeatedly applying the addition formulas, one finds that both $\hat{A}_{i}$ and $\hat{B}_{i}$ may be expressed in the form $d/ |\prod_{j \neq i} a_j|$ for some $d \in \Z$ and hence the claim follows.\\
Returning to the proof of \eqref{Sbound}, we set
$$\beta:=\frac{k-1}{k}\geq \frac{1}{2}$$
and consider two cases.

\begin{itemize}
\item[] {\bf Case I}: both $|\hat{A}_{k}|> y^{- \beta \delta}$ and $|\hat{B}_{k}|> y^{- \beta \delta}$.\\
\end{itemize}

Let $U_{\delta}^{+}=U_{\delta}^{+}(y; k, \epsilon, \underline{\eta}^{+}, \underline{\nu}^{+})$ denote the restriction of $S_{\delta}^{+}$ in which the summation variables $a_1,...,a_k$ satisfy  \eqref{indfix} as well as the hypothesis in Case I. Given any pair $(\hat{A}_{k}, \hat{B}_{k})$ we can use (\ref{separatek}) to rewrite the condition (\ref{conds})(+) as
$$0<| r \cos(\theta_k -\alpha)|=|\hat{A}_{k} \cos(\theta_k) +\hat{B}_{k} \sin(\theta_k) |< y^{-\delta}, $$
where
$$r=(\hat{A}_{k}^2+ \hat{B}_{k}^2)^{1/2}, \ \ \tan \alpha=\frac{\hat{B}_{k}}{\hat{A}_{k} }.$$
As a result we have that $|\cos(\theta_k - \alpha)| \leq y^{-\delta/k}/\sqrt{2}$. From the mean value theorem (applied to the function $\arccos(x) $), it follows that $\theta_k$ must live in one of two intervals $I_k, I_k'$ each having length at most $2y^{-\delta/k}$. Recalling that $|a_k| \geq y^{\delta/ k}$ \qcnew{and, if necessary, expanding the intervals $I_k$, $I_k'$ to be of length exactly $2y^{-\delta/k}$}, we may now apply Lemmas \ref{sectorbound} and \ref{exppowersums} to get
\begin{align}\label{U}
U_{\delta}^{+}(y; k, \epsilon, \underline{\eta}^{+}, \underline{\nu}^{+}) & \leq  \   \sum_{ \substack{a_1,...,a_{k-1} \in \Z[i] \\  |\prod_j a_j | \leq 2 y^{1-\delta/k} } }
 \ \sum_{  \substack{     |a_k| \asymp y/|\prod_j a_j |  \\ \theta_k \in I_k \cup I_k'  }    }  1 \notag \\
&\leq 2 \kappa_1   \sum_{ \substack{a_1,...,a_{k-1} \in \Z[i] \\  |a_j | \leq 2 y \ \forall j} }  \left(  \frac{2y}{|\prod_{j \leq k-1} a_j | }\right)^{2 } (2y^{-\delta/k})  \\ \notag
& \leq 16 \kappa_1 ( \kappa_2)^{k-1} y^{2-\delta/k} (\log (2+2y))^{k-1} \leq  \frac{1}{2} (\kappa_3)^{k}  y^{2-\delta/4^k } (\log (2+2y))^{k-1}.
\end{align}
In the last step we used the inequality $16 \leq \kappa_3/2$.

\begin{itemize}
\item[] {\bf Case II}: either $0<|\hat{A}_{k}|\leq y^{-\beta \delta}$ or $0<|\hat{B}_{k}|\leq y^{- \beta \delta}$.\\
\end{itemize}

Let $V_{\delta}^{+}=V_{\delta}^{+}(y; k, \epsilon, \underline{\eta}^{+}, \underline{\nu}^{+})$ denote the restriction of $S_{\delta}^{+}$ according to \eqref{indfix} and the hypothesis in Case II. By the claim we have the lower bound
$$\left|\prod_{j \leq k-1} a_j\right| \geq y^{\beta \delta}$$ for any tuple $a_1,...,a_k$ appearing in $V_{\delta}^{+}$ and hence we may assume \qcnew{that} one of the $a_i$ is large, say
\begin{equation}\label{largea1}
|a_1| \geq y^{\delta/k}.
\end{equation}
We will consider subcases II(a) and II(b) and write $V_{\delta}^{+}(a)$ and $V_{\delta}^{+}(b)$ for the corresponding restrictions of $V_{\delta}^{+}$.

\begin{itemize}
\item[] {\bf Case II(a)}:  $|a_{k}|\leq 2y^{1/2}$.\\
\end{itemize}

In this situation we apply the induction hypothesis to get
\begin{align*}
V_{\delta}^{+}(a) & \leq \sum_{ |a_k| \leq 2 y^{1/2} }  \    \sum_{ \substack{a_1,...,a_{k-1} \in \Z[i] \\  |\prod_j a_j | \asymp y/|a_k| }}^{ \text{case II} } 1\\
&\leq 2 (\kappa_3)^{k-1} (2k-2)! \ (\log (2+2y))^{k-2} \sum_{ |a_k| \leq 2 y^{1/2} } \left(  \frac{y}{|a_k|}\right)^{2-\beta \delta/4^{k-1} }.
\end{align*}
Applying the straightforward inequality $4 \kappa_2 \leq \kappa_3/2$ together with Lemma \ref{exppowersums}  we get
\begin{align}\label{Va}
 V_{\delta}^{+}(a) & \leq  2 (\kappa_3)^{k-1} (2k-2)! \ (\log (2+2y))^{k-2}  \notag \\
& \Large\times \kappa_2 \  y^{2-\beta \delta/4^{k-1} }\cdot (2y^{1/2} )^{ \beta \delta/4^{k-1} }  \log (2+2y)  \\ \notag
& \leq \frac{1}{2}(\kappa_3)^k (2k-2)!  \ y^{2-\delta/4^{k}  }  (\log (2+2y))^{k-1}.
\end{align}

\begin{itemize}
\item[] {\bf Case II(b)}:  $|a_{1}|\leq 2y^{1/2}$.\\
\end{itemize}
In this scenario, interchanging the roles of $a_1$ and $a_k$, we satisfy the criteria of either Case I or Case II(a) (note that \eqref{largea1} is necessary to obtain the estimates carried out in Case I).  It follows that $V_{\delta}^{+}(b) \leq V_{\delta}^{+}(a) + U_{\delta}^{+}$.\\
Collecting the estimates from the two subcases we find that
\begin{equation}\label{V}
V_{\delta}^{+}\leq (k-1) (V_{\delta}^{+}(a) + V_{\delta}^{+}(b)   )\leq (k-1) (2V_{\delta}^{+}(a) + U_{\delta}^{+}   ),
\end{equation}
where the extra factor $k-1$ compensates for the loss we incurred by fixing the index in \eqref{largea1}.

To conclude the proof of Proposition \ref{Spm} we combine (\ref{U}), (\ref{Va}) and \eqref{V} to get
\begin{align*}
S^{+}_{\delta} \leq k (U^{+}_{\delta} +V^{+}_{\delta})
 \leq k^2 U^{+}_{\delta} + 2k(k-1) V_{\delta}^{+}(a)   \leq   (\kappa_3)^{k} (2k)!  y^{2-\delta/4^{k} } (\log (2+2y))^{k-1}.
\end{align*}
As before, the factor $k$ in the first inequality compensates for fixing the index in \eqref{indfix}.
\end{proof}
\end{proposition}

\subsubsection{Concluding the proof of Theorem \ref{maincorr}a.}\label{complcorr}

\begin{proof}
It remains to estimate the RHS of (\ref{eset}). Let $\underline{\eta}^{+}, \underline{\nu}^{+}, \epsilon$ be fixed and consider any ascending $k$-tuple of Gaussian primes  $|\pi_1|\leq ...\leq |\pi_k|  $ which live in the first quadrant and satisfy the condition (\ref{conds})(+). For each $j \leq k$ the vector $(\epsilon_j^{(r)})_{r \leq l}$ corresponds to one of the $2^{l}$ elements of $\left\{-1,1 \right\}^{l}$. We will regroup all of the $j \leq k$ which give rise to the same element of $\left\{-1,1 \right\}^{l}$. To this end, let $\tau_1,...,\tau_L$ be a list of all the elements in $\left\{-1,1 \right\}^{l}$ (so that $L=2^{l}$) and define
$$a_i:=\prod_{(\epsilon_j^{(r)})_{r \leq l} = \tau_i} \pi_j, \qquad i=1,...,L,$$
with the convention that the empty product is $1$. In this manner we obtain a map
$$\phi^{+}: \left\{ |\pi_1|\leq ...\leq |\pi_k| \right\} \longmapsto \left\{ a_1,..., a_{L } \right\}.$$
Some remarks are in order.\\
First, the map $\phi^{+}$ is injective since we only consider $k$-tuples of primes in ascending order. Second, any tuple of Gaussian integers $a_1,..., a_{ L}$ in the image of $\phi^{+}$ will also satisfy the condition (\ref{conds})(+) for some matrix $  \tilde{\epsilon}=(\tilde{ \epsilon}^{\ (r)}_j )_{1\leq r \leq {l}, 1\leq j\leq L} $ and the same choice of $\underline{\eta}^{+}, \underline{\nu}^{+}$.\\
In light of (\ref{eset}) and Proposition \ref{Spm} we now find that for $K \ll \log N$,
\begin{align*}
|\mathcal{R}_N(l; \delta)| & \ll \sum_{\substack{ \underline{\eta}^{+} \in \left\{-1,1 \right\}^{l}\\  \underline{\nu}^{+} \in [4]^{l}}      }   \     \sum_{k < K} 2^{k l} \ \max_{\tilde{\epsilon}= \tilde{\epsilon}^{ (r)}_j  } \left| S_{\delta}^{+}(Y; L, \tilde{\epsilon}, \underline{\eta}^{+}, \underline{\nu}^{+}) \right|  + O\left(    \frac{N (\log N)^5}{2^{K}}   \right) \\
& \ll (\kappa_3)^{L  } 8^{l} (2L)! \ N^{1-\rho(\delta, l) } (\log (2+2\sqrt{N}))^{L -1} 2^{K l}+   \frac{N( \log N)^5}{2^{K}}   ,
\end{align*}
with $\rho(\delta, l)= \delta/(2 \cdot 4^{L}) $. Noting that $\log (2+2\sqrt{N})\leq \log N$ we now set $\kappa:=\qcnew{16} \kappa_3$ and choose $2^K \asymp N^{\rho/ (l +1)}$ \qcnew{(so that $2^{K \ell} \leq 2^{l} N^{\rho l / (l +1)}   $)} to get

\begin{align*}
|\mathcal{R}_N(l; \delta)| &\ll (\qcnew{16} \kappa_3)^{L } (2L)! \ N^{1-\rho/(l +1)} (\log N)^{L -1}+ N^{1-\rho/(l +1)}  ( \log N)^5 \\
& \ll \kappa^{L  }  (2L)! \ N^{1-\rho/(l +1)} (\log N)^{L +1},
\end{align*}
as claimed.

\end{proof}

}

\qnew{

\subsection{A sequence $(n_k)_{k \geq 1}$ for which $\tau_{n_k} \Rightarrow  \nu_s$}

Let $s \in [0, \pi/4]$ and take $N$ to be a large natural number. We choose $k$ to be the integer for which $2^k \asymp \log N$ and select from $[0,s]$ the disjoint subintervals
\begin{equation}\label{alphabetadef}
[\alpha_j, \beta_j]:=\left[ \frac{s2^j}{ 2^k} ,  \frac{s 2^j}{ 2^k} \left(1+ \frac{1}{ k^2} \right) \right], \qquad j=0,...,k-1.
\end{equation}
\begin{lemma}\label{Boolean}
For any tuple $(\theta_0,...,\theta_{k-1}) \in [\alpha_0, \beta_0]\times... \times [\alpha_{k-1}, \beta_{k-1}]$ the sumset
$$\mathcal{B}(\theta_0,...,\theta_{k-1}):=\left\{ \sum_{j=0}^{k-1} \epsilon_j \theta_j :  \ \epsilon_j \in \left\{-1,1 \right\}  \qcnew{ \text{ for all $j$ } } \right\}$$
forms a collection of $(s/2^k)$-separated numbers. Each of the intervals
$$\left[\frac{s(2j+1)(1-\frac{1}{k})     }{2^k}, \frac{s(2j+1) (1+\frac{1}{k})}{2^k}   \right], \qquad (-2^{k-1}\leq j \leq 2^{k-1}-1)$$
contains exactly one member of $\mathcal{B}(\theta_0,...,\theta_{k-1})$.
\begin{proof}
For each $0\leq i \leq k-1$ we may write $\theta_i= s2^i/2^k + \gamma_i$ with $|\gamma_i| \leq s2^i/(k^2 2^k)$. Since the sumset $\mathcal{B}(1,2,...,2^{k-1})$ consists of all odd integers in $[-2^k,2^k]$, the result follows.
\end{proof}
\end{lemma}

We now introduce the set of $k$-almost Gaussian primes
$$\mathcal{G}_k(N;\underline{\alpha}, \underline{\beta}) : =\left\{a=\prod_{j \leq k} \pi_j: \ |a|\leq N^{1/2}, \ \ \pi_j \in \Gamma(N^{1/2},\alpha_j,\beta_j)  \qcnew{ \text{ for all $j$ } } \right\},$$
where the angles $\alpha_i, \beta_i$ are chosen as in \eqref{alphabetadef} (and hence depend on $N$).
$\mathcal{G}_k$ gives rise to the relatively large set of rational integers
$$\mathcal{A}_k(N;\underline{\alpha}, \underline{\beta}) := \left\{n \leq N: n=|a|^2,  a \in \mathcal{G}_k(N;\underline{\alpha}, \underline{\beta})  \right\}$$
from which we will extract the sequence $(n_k)_k$. Indeed, since $k=O(\log_2 N)$ we may apply Lemma \ref{almostprimes} to get
$$|\mathcal{A}_k(N;\underline{\alpha}, \underline{\beta})|=|\mathcal{G}_k(N;\underline{\alpha}, \underline{\beta})| \geq \left( \frac{ c }{\qcnew{4}}\right)^k \frac{N}{ \qcnew{(\log N)^{3k}  } } \gg \frac{N}{\Phi(N)}.$$
Comparing this lower bound to the estimate \eqref{RN} for $|\mathcal{R}_N(l; \delta)| $ we deduce the existence of an infinite sequence \qcnew{ $(n_k)_{k \geq1}$ for which
$$ \mathcal{C}_{n_k}(l;n^{(1-\delta)/2}) =\emptyset. $$
 }

\begin{proof}[Proof of Theorem \ref{maincorr}b.]

Let $f:S^{1}\rightarrow \R$ be an arbitrary continuous test function and consider any $n_k=|\prod_{i \leq r} \pi_i|^2$ belonging to the sequence introduced \qcnew{just above}. As before we let $\theta_i$ denote the angle between $\pi_i$ and the x-axis and observe that
$$\frac{\mathcal{E}_{n_k}  }{\sqrt{n_k}}=\left\{e^{ix }:  x \in \bigcup_{j=0}^{3}  \left(j \pi/2 +\mathcal{B}(\theta_1,..., \theta_r) \right) \right\}.$$
It follows from Lemma \ref{Boolean} that the sumset $\mathcal{B}(\theta_1,..., \theta_r)=:\left\{x_1,...,x_J \right\}$ is evenly distributed in $[-s,s]$: we have that $x_{j+1}-x_j=2s/J + o_r(1)$ for all $j\leq J-1$.
Recalling Definition \ref{measdef} we now get
\begin{equation}\label{mun}
\int\limits_{\Sc^{1}} f \ d \tau_{n_k}= \frac{1}{4 |\mathcal{B}(\theta_1,..., \theta_r)|} \sum_{j=0}^3 \  \sum_{x \in \mathcal{B}(\theta_1,..., \theta_r)} f(e^{i (x+ j \pi/2)}).
\end{equation}
The RHS of \eqref{mun} represents, up to a small error, an evenly spaced Riemann sum for the integral $\frac{1}{8s} \sum_{j=0}^3 \int_{-s+k \pi/2}^{s+k \pi/2} f $. By construction, the size of $\mathcal{E}_{n_k}$ will grow together with $n_k$ and hence $\tau_{n_k} \Rightarrow \nu_s$.

\end{proof}

}

\appendix

\section{Proof of Lemma \ref{lem:r der moments}: moments of $r$ and its derivatives along the shrinking balls $B(s)$}

\label{apx:r der moments}

The goal of this section is to prove Lemma \ref{lem:r der moments}. First we need to formulate
the following lemma whose purpose is evaluating some summations of oscillatory integrals;
it will be proven after the proof of Lemma \ref{lem:r der moments}.

\begin{lemma}

\label{lem:sum pairs osc int}

\begin{enumerate}

\item
For every $s=s(n)>0$ and $K=K(n)>0$ we have the estimate
\begin{equation}
\label{eq:sum osc int B(s) pairs}
\sum\limits_{\lambda\ne\lambda'\in\Ec_{n}}\left|\int\limits_{B(s)}e(\langle \lambda-\lambda',x\rangle)dx\right|^{2} \ll
s^{4}\cdot\left(|\Cc_{n}(2;K)|+ \frac{\Nc_{n}^{2}}{(Ks)^{3}}\right)
\end{equation}

\item
For every $s=s(n)>0$ and $K=K(n)>0$ we have the estimate
\begin{equation}
\label{eq:sum osc int B(s) 4-tuples non-diag}
\sum\limits_{\lambda+\lambda'+\lambda''+\lambda'''\ne 0}
\left|\int\limits_{B(s)}e(\langle \lambda+\lambda'+\lambda''+\lambda''',x\rangle)dx\right|^{2}
\ll s^{4}\cdot\left(|\Cc_{n}(4;K)|+ \frac{\Nc_{n}^{4}}{(Ks)^{3}}\right)
\end{equation}

\item
For every $s=s(n)>0$ and $K=K(n)>0$ we have the estimate
\begin{equation}
\label{eq:sum osc int B(s) 6-tuples non-diag}
\sum\limits_{\lambda_{1}+\ldots+\lambda_{6}\ne 0}
\left|\int\limits_{B(s)}e(\langle \lambda_{1}+\ldots+\lambda_{6},x\rangle)dx\right|^{2}
\ll s^{4}\cdot\left(|\Cc_{n}(6;K)|+ \frac{\Nc_{n}^{6}}{(Ks)^{3}}\right)
\end{equation}

\end{enumerate}

\end{lemma}

\begin{proof}[Proof of Lemma \ref{lem:r der moments} assuming Lemma \ref{lem:sum pairs osc int}]

First we prove \eqref{eq:rn 2nd mom Bs}; we then assume that $n\in S'\subseteq S$ satisfies the hypothesis
\begin{equation}
\label{eq:Ac(n,2,n^1/2-delt assume}
\Ac(n;2,n^{1/2-\delta}),
\end{equation}
and aim at proving \eqref{eq:rn 2nd mom Bs} for an arbitrary ball of radius satisfying
\begin{equation}
\label{eq:s>n^-1/2+eps}
s>n^{-1/2+\epsilon}.
\end{equation}
Using the definition \eqref{eq:rn covar def} of $r_{n}$ and separating the diagonal contribution we have
\begin{equation}
\label{eq:rn^2 sum}
r_{n}(x)^{2} = \frac{1}{\Nc_{n}} + \frac{1}{\Nc^{2}}\sum\limits_{\lambda\ne\lambda'\in\Ec_{n}} e(\langle \lambda-\lambda',x\rangle),
\end{equation}
hence
\begin{equation}
\label{eq:2nd moment diag sep err}
\begin{split}
&\int\limits_{B(s)\times B(s)} r_{n}(x-y)^{2}dxdy = (\pi s^{2})^2\cdot \frac{1}{\Nc_{n}} +
\frac{1}{\Nc^{2}}\sum\limits_{\lambda\ne\lambda'\in\Ec_{n}} \int\limits_{B(s)\times B(s)} e(\langle \lambda-\lambda',x-y\rangle)dxdy
\\&= (\pi s^{2})^2\cdot \frac{1}{\Nc_{n}} + \frac{1}{\Nc^{2}}\sum\limits_{\lambda\ne\lambda'\in\Ec_{n}}\left|\int\limits_{B(s)}e(\langle \lambda-\lambda',x\rangle)dx\right|^{2}
\\&= (\pi s^{2})^2\cdot \frac{1}{\Nc_{n}} + O\left( s^{4}\cdot\left(|\Cc_{n}(2;K)|+ \frac{\Nc_{n}^{2}}{(Ks)^{3}}\right) \right),
\end{split}
\end{equation}
by Lemma \ref{lem:sum pairs osc int};
we still maintain the freedom of choosing the threshold $K=K(n)$, with the help of \eqref{eq:s>n^-1/2+eps}.

For the choice $$K=n^{1/2-\delta}$$ we have that the quasi-correlation set
\begin{equation}
\label{eq:Ccn(2,M)=empty}
\Cc_{n}(2;K)=\varnothing
\end{equation}
is empty by hypothesis \eqref{eq:Ac(n,2,n^1/2-delt assume} we made earlier in this proof, and
\begin{equation}
\label{eq:N^2/(Ms)^3 small}
\frac{\Nc_{n}^{2}}{(Ks)^{3}} = O\left(\frac{\Nc_{n}^{2}}{n^{\epsilon-\delta}}\right) = O\left(\frac{1}{\Nc_{n}^{A}}\right)
\end{equation}
is smaller than any power $A>0$ of $\Nc_{n}$, bearing in mind \eqref{eq:delt<eps} and \eqref{eq:Nc dim << n^eps}.
Upon substituting the last couple of estimates,
\eqref{eq:Ccn(2,M)=empty} and \eqref{eq:N^2/(Ms)^3 small}
 into \eqref{eq:2nd moment diag sep err} we then obtain the asymptotics
\begin{equation*}
\int\limits_{B(s)\times B(s)} r_{n}(x-y)^{2}dxdy = (\pi s^{2})^2\cdot \frac{1}{\Nc_{n}}\left(1+O\left( \frac{1}{\Nc_{n}^{A}} \right) \right)
\end{equation*}
holding for every $A>0$ i.e. we obtain \eqref{eq:rn 2nd mom Bs}. Along the way we also proved that, under the assumptions of
Lemma \ref{lem:r der moments}, the summation
\begin{equation}
\label{eq:sum osc 2nd mom << N^-A}
\sum\limits_{\lambda\ne\lambda'\in\Ec_{n}}\left|\int\limits_{B(s)}e(\langle \lambda-\lambda',x\rangle)dx\right|^{2}
\ll s^{4}\cdot \frac{1}{\Nc_{n}^{A}}
\end{equation}
is smaller than any power of $\Nc_{n}^{A}$.

\vspace{3mm}

We turn now to proving \eqref{eq:rn 4th mom Bs}. First, under the hypothesis $\Ac(n;2,n^{1/2-\delta/2})$, also
\begin{equation}
\label{eq:assume Ac(n;4,n^1/2-delta)}
\Ac(n;4,n^{1/2-\delta})
\end{equation}
is satisfied, thanks to Lemma \ref{lem:RW qcorr2=>qcorr4}.
Similarly to \eqref{eq:rn^2 sum} and \eqref{eq:2nd moment diag sep err},
and upon recalling the definition \eqref{eq:Dc diag def} of the diagonal, and that the length-$4$ spectral correlation
set $\Sc_{n}(4)$ consists of the diagonal elements only \eqref{eq:Dc4=Sc4}, for the $4$th moment we have
\begin{equation*}
r_{n}(x)^{4} = \frac{|\Dc_{n}(4)|}{\Nc_{n}^{4}} +
\frac{1}{\Nc^{4}}\sum\limits_{\lambda\ne\lambda'\in\Ec_{n}} e(\langle \lambda+\lambda'+\lambda''+\lambda''',x\rangle),
\end{equation*}
and
\begin{equation*}
\begin{split}
&\int\limits_{B(s)\times B(s)} r_{n}(x-y)^{4}dxdy = (\pi s^2)^{2} \cdot \frac{|\Dc_{n}(4)|}{\Nc_{n}^{4}} \\&+
\frac{1}{\Nc^{4}}\sum\limits_{\lambda+\lambda'+\lambda''+\lambda'''\ne 0}
\left|\int\limits_{B(s)}e(\langle \lambda+\lambda'+\lambda''+\lambda''',x\rangle)dx\right|^{2}.
\end{split}
\end{equation*}
Hence, by invoking Lemma \ref{lem:sum pairs osc int} once again, we obtain
\begin{equation}
\label{eq:4th moment diag sep err}
\int\limits_{B(s)\times B(s)} r_{n}(x-y)^{4}dxdy = (\pi s^2)^{2} \cdot \frac{|\Dc_{n}(4)|}{\Nc_{n}^{4}} +
O\left( s^{4}\cdot\left(|\Cc_{n}(4;K)|+ \frac{\Nc_{n}^{4}}{(Ks)^{3}}\right) \right),
\end{equation}
where we are still free to choose the value of the parameter $K=K(n)$.

For the choice $K=n^{1/2-\delta}$ as above, the length-$4$
correlation set $\Cc_{n}(4;K)$ is empty by \eqref{eq:assume Ac(n;4,n^1/2-delta)}, and
\begin{equation*}
\frac{\Nc_{n}^{4}}{(Ks)^{3}} = O\left(\frac{\Nc_{n}^{4}}{n^{\epsilon-\delta}}\right) = O\left(\frac{1}{\Nc_{n}^{A}}\right)
\end{equation*}
for every $A>0$ by \eqref{eq:Nc dim << n^eps}, and \eqref{eq:4th moment diag sep err} now reads
\begin{equation*}
\int\limits_{B(s)\times B(s)} r_{n}(x-y)^{4}dxdy = (\pi s^2)^{2} \cdot \frac{|\Dc_{n}(4)|}{\Nc_{n}^{4}}\left(1 +
O\left(\frac{1}{\Nc_{n}^{A}}\right) \right),
\end{equation*}
i.e. the first estimate of \eqref{eq:rn 4th mom Bs}. The second estimate of \eqref{eq:rn 4th mom Bs}
follows from the first one and
\begin{equation*}
|\Dc_{n}(4)| = 3\Nc_{n}^{2}+O\left(\Nc_{n}\right).
\end{equation*}
On the way we also proved that, under the assumptions of
Lemma \ref{lem:r der moments}, the summation
\begin{equation}
\label{eq:sum osc 4th mom << N^-A}
\sum\limits_{\lambda+\lambda'+\lambda''+\lambda'''\ne 0}
\left|\int\limits_{B(s)}e(\langle \lambda+\lambda'+\lambda''+\lambda''',x\rangle)dx\right|^{2}
\ll s^{4}\cdot \frac{1}{\Nc_{n}^{A}}
\end{equation}
is smaller than any power of $\Nc_{n}^{A}$.

Now we turn to proving \eqref{eq:rn 6th mom Bs} under
\begin{equation}
\label{eq:A(n,6,n^1/2-delta) assumed}
\Ac(n;6,n^{1/2-\delta/2}).
\end{equation}
As above, we have
\begin{equation}
\label{lem:sum 6tuples osc int}
\int\limits_{B(s)\times B(s)} r_{n}(x-y)^{6}dxdy = (\pi s^{2})^{2} \cdot \frac{|\Sc_{n}(6)|}{|\Nc_{n}|^{6}}
+\frac{1}{\Nc^{4}}\sum\limits_{\lambda_{1}+\ldots+\lambda_{6}\ne 0}
\left|\int\limits_{B(s)}e(\langle \lambda_{1}+\ldots+\lambda_{6},x\rangle)dx\right|^{2}.
\end{equation}
For the first summand on the r.h.s. of \eqref{lem:sum 6tuples osc int} we use Bombieri-Bourgain's bound \new{\eqref{eq:S6<<N^7/2}}, whereas
for the second summand on the r.h.s. of \eqref{lem:sum 6tuples osc int} we invoke Lemma \ref{lem:sum pairs osc int}. These yield
\begin{equation*}
\int\limits_{B(s)\times B(s)} r_{n}(x-y)^{6}dxdy = s^{4}\cdot \left(O\left(\frac{1}{\Nc_{n}^{5/2}}\right) +
|\Cc_{n}(6;K)|+ \frac{\Nc_{n}^{6}}{(Ks)^{3}} \right).
\end{equation*}
The estimate \eqref{eq:rn 6th mom Bs} then finally follows upon choosing $K=n^{1/2-\delta}$ so that $$\Cc_{n}(6;K)=\varnothing$$
by \eqref{eq:A(n,6,n^1/2-delta) assumed}, and
\begin{equation*}
\frac{\Nc_{n}^{6}}{(Ks)^{3}} \ll \frac{\Nc_{n}^{6}}{n^{\epsilon}} \ll \frac{1}{\Nc_{n}^{A}}
\end{equation*}
for every $A>0$, by \eqref{eq:delt<eps} and \eqref{eq:Nc dim << n^eps}.

\vspace{2mm}

Now we turn to proving \eqref{eq:mom tr(X)}; the proof is quite similar to the proof of \eqref{eq:rn 2nd mom Bs} except that we need
to deal with the potentially blowing up denominator in \eqref{eq:XY block def}. To this end we separate the singular set $B_{sing}$
defined in \S\ref{sec:singular set}. The contribution of $B_{sing}$ to the l.h.s. of \eqref{eq:mom tr(X)} is
\begin{equation*}
\int\limits_{(B(s)\times B(s))\cap B_{sing}}\tr{X_{n}(x-y)}dxdy \ll \meas(B_{sing}) \ll \Rc_{n}(6;2s) \ll s^{4}\cdot \frac{1}{\Nc_{n}^{5/2}},
\end{equation*}
by \eqref{eq:X,Y=O(1)}, \eqref{eq:meas(sing)<<R6 mom} and \eqref{eq:rn 6th mom Bs} respectively.
On $(B(s)\times B(s))\setminus B_{sing}$ we may expand
\begin{equation*}
\frac{1}{1-r_{n}^{2}} = 1+r_{n}^{2}+O(r_{n}^{4}),
\end{equation*}
hence
\begin{equation}
\label{eq:mom tr(X) as DDt,r^2DDt}
\begin{split}
&\int\limits_{(B(s)\times B(s))\setminus B_{sing}}\tr{X_{n}(x-y)}dxdy \\&=
-\frac{2}{E_{n}}\int\limits_{(B(s)\times B(s))\setminus B_{sing}}
\left(\tr(D^{t}D)+r_{n}^{2}\tr(D^{t}D) + O(r_{n}^{4}\tr(D^{t}D))\right)dxdy
\\&=-\frac{2}{E_{n}}\int\limits_{B(s)\times B(s)}
\left(\tr(D^{t}D)+r_{n}^{2}\tr(D^{t}D) + O(r_{n}^{4}\tr(D^{t}D))\right)dxdy
+ O(\meas(B_{sing})),
\end{split}
\end{equation}
since $r_{n}$ and $\frac{D^{t}}{\sqrt{n}}$ are both uniformly bounded.

Now by \eqref{eq:D grad def, sum lat} we have upon separating the diagonal $\lambda=\lambda'$
\begin{equation}
\label{eq:mom trD^tD asymp}
\begin{split}
&\int\limits_{B(s)\times B(s)}\tr(D^{t}\cdot D)dxdy = \int\limits_{B(s)\times B(s)}D\cdot D^{t}dxdy
\\&= \new{(\pi s^{2})^{2}\cdot 4}\pi^{2}n\cdot \frac{1}{\Nc_{n}}+ O \left(n\cdot \sum\limits_{\lambda\ne\lambda'\in\Ec_{n}}\left|\int\limits_{B(s)}e(\langle \lambda-\lambda',x\rangle)dx\right|^{2}  \right)
\\&= (\pi s^{2})^{2} \cdot E_{n}\cdot \frac{1}{\Nc_{n}} + \new{O\left(s^{4}\cdot n\frac{1}{\Nc_{n}^{A}}\right)}
\end{split}
\end{equation}
\new{for every $A>0$}, by \eqref{eq:En=4pi^2n} and  \eqref{eq:sum osc 2nd mom << N^-A}. For the second summand on the r.h.s. of \eqref{eq:mom tr(X) as DDt,r^2DDt}
we have, again separating the diagonal $\lambda+\lambda'+\lambda''+\lambda'''=0$:
\begin{equation}
\label{eq:mom r^2trD^tD asymp}
\begin{split}
&\int\limits_{B(s)\times B(s)}r_{n}^{2}\tr(D^{t}D) dxdy =  \int\limits_{B(s)\times B(s)}r_{n}^{2}DD^{t}dxdy
\\&= (\pi s^{2})^{2} \cdot E_{n}\cdot \new{\frac{1}{\Nc_{n}^{2}}} + O\left(E_{n}\sum\limits_{\lambda+\lambda'+\lambda''+\lambda'''\ne 0}
\left|\int\limits_{B(s)}e(\langle \lambda+\lambda'+\lambda''+\lambda''',x\rangle)dx\right|^{2}\right)
\\&=(\pi s^{2})^{2} \cdot \new{E_{n}\cdot \frac{1}{\Nc_{n}^{2}}} + O\left(s^{4}E_{n}\cdot \frac{1}{\Nc_{n}^{A}}    \right),
\end{split}
\end{equation}
by \eqref{eq:sum osc 4th mom << N^-A}. For the latter term on the r.h.s. of \eqref{eq:mom tr(X) as DDt,r^2DDt} we
employ \eqref{eq:D grad def, sum lat}, so that
\begin{equation}
\label{eq:mom r^4trD^tD bnd}
\begin{split}
&\int\limits_{B(s)\times B(s)}r_{n}^{4}\tr(D^{t}D) dxdy \\&\ll
E_{n}\cdot\left( (\pi s^{2})^{2} \cdot \frac{|\Sc_{n}(6)|}{|\Nc_{n}|^{6}}
+\frac{1}{\Nc^{4}}\sum\limits_{\lambda_{1}+\ldots+\lambda_{6}\ne 0}
\left|\int\limits_{B(s)}e(\langle \lambda_{1}+\ldots+\lambda_{6},x\rangle)dx\right|^{2} \right)
\\&= E_{n}\cdot \int\limits_{B(s)\times B(s)} r_{n}(x-y)^{6}dxdy \ll
O\left(E_{n} s^4\cdot \frac{1}{\Nc_{n}^{5/2}}\right)
\end{split}
\end{equation}
by \eqref{lem:sum 6tuples osc int} and \eqref{eq:rn 6th mom Bs}.
The estimate \eqref{eq:mom tr(X)} finally follows upon substituting \eqref{eq:mom trD^tD asymp},
\eqref{eq:mom r^2trD^tD asymp} and \eqref{eq:mom r^4trD^tD bnd} into \eqref{eq:mom tr(X) as DDt,r^2DDt}.
The exact same argument used for \eqref{eq:mom tr(X)} also yields \eqref{eq:mom tr(Y^2)}.

\vspace{3mm}

The proofs of all the other estimates \eqref{eq:mom tr(XY^2)}-\eqref{eq:mom tr(Y^6)} also follow the same but slightly easier pattern:
we bound the contribution of the singular set $(B(s)\times B(s))\cap B_{sing}$ using the uniform boundedness
of the integrand, and expand $$\frac{1}{1-r^{2}}=1+O(r^{2})$$
(no need for $r^{2}$) on $(B(s)\times B(s))\setminus B_{sing}$. The main contribution
will always come from evaluating the appropriate moments of $r$ on the diagonal $\Dc_{n}(4)$, whereas the non-diagonal
contribution is bounded by \eqref{eq:sum osc 2nd mom << N^-A}, \eqref{eq:sum osc 4th mom << N^-A}; the term corresponding to $O(r^{2})$
is bounded using \eqref{eq:rn 6th mom Bs}. The precise details are omitted here (cf. ~\cite[Proofs of lemmas 4.6, 5.4]{K-K-W}).

\end{proof}

\begin{proof}[Proof of Lemma \ref{lem:sum pairs osc int}]

First, we show \eqref{eq:sum osc int B(s) pairs}. To this end we transform the variables $s\cdot y=x$ to write
\begin{equation}
\label{eq:int B(s) e Fourier chi}
\int\limits_{B(s)} e(\langle \lambda-\lambda',x\rangle)dx = s^{2}\int\limits_{B(1)}e(\langle s(\lambda-\lambda'),y\rangle)dy
=s^{2}\widehat{\chi}(s(\lambda-\lambda')),
\end{equation}
where $\chi=\chi_{B(1)}$ is the characteristic function of the Euclidian unit ball $B(1)\subseteq \R^{2}$. As $\chi$ is rotationally
invariant so is its Fourier transform, and a standard direct computation shows that
\begin{equation}
\label{eq:chihat(xi)via J1}
\widehat{\chi}(\xi) = 2\pi\frac{J_{1}(\|\xi\|)}{\|\xi\|},
\end{equation}
where $J_{1}$ is the Bessel $J$ function; it is well known that
\begin{equation}
\label{eq:J1<<1,1/sqrt(psi)}
J_{1}(\psi) \ll \min\left\{1, \frac{1}{\|\psi\|^{1/2}}\right\}.
\end{equation}

Upon substituting \eqref{eq:J1<<1,1/sqrt(psi)} into \eqref{eq:chihat(xi)via J1}, and then into \eqref{eq:int B(s) e Fourier chi} we obtain
\begin{equation}
\label{eq:int B(s) exp << 1,1/sqrt(psi)}
\left|\int\limits_{B(s)} e(\langle \lambda-\lambda',x\rangle)dx\right| \ll
s^{2}\cdot \min\left\{1, \frac{1}{(s\|\lambda-\lambda'\|)^{3/2}}\right\}.
\end{equation}
We then use the inequality \eqref{eq:int B(s) exp << 1,1/sqrt(psi)} to bound the summands of \eqref{eq:sum osc int B(s) pairs},
separating the contribution of the range $\|\lambda-\lambda'\|>K$ and otherwise
to yield
\begin{equation*}
\sum\limits_{\lambda\ne\lambda'\in\Ec_{n}}\left|\int\limits_{B(s)}e(\langle \lambda-\lambda',x\rangle)dx\right|^{2} \ll
s^{4}\left( \left|\{(\lambda,\lambda'):\: 0<|\lambda-\lambda'|\le K\} \right| + \sum\limits_{\|\lambda-\lambda'\|>K} \frac{1}{(sK)^{3}}\right),
\end{equation*}
which implies \eqref{eq:sum osc int B(s) pairs} upon recalling the definition \eqref{eq:Cc qcorr def} of the quasi-correlation set $\Cc(2;K)$,
and trivially bounding the number of summands in the latter summation by $\Nc^{2}$.

Now we turn to proving \eqref{eq:sum osc int B(s) 4-tuples non-diag}. Following along the lines of the above argument for
\eqref{eq:sum osc int B(s) pairs} we find that the l.h.s. of \eqref{eq:sum osc int B(s) 4-tuples non-diag} is equal to
\begin{equation*}
\begin{split}
&\sum\limits_{\lambda+\lambda'+\lambda''+\lambda'''\ne 0}
\left|\int\limits_{B(s)}e(\langle \lambda+\lambda'+\lambda''+\lambda''',x\rangle)dx\right|^{2}
\\&= 4\pi^{2}s^{4}\sum\limits_{\lambda+\lambda'+\lambda''+\lambda'''\ne 0}\frac{J_{1}(s\|\lambda+\lambda'+\lambda''+\lambda'''\|)^{2}}{(s\|\lambda+\lambda'+\lambda''+\lambda'''\|)^{2}},
\end{split}
\end{equation*}
which is bounded by separating the summands $\|\lambda+\lambda'+\lambda''+\lambda'''\| > K$ from the other summands
and using the respective inequalities from \eqref{eq:J1<<1,1/sqrt(psi)} in each of the cases. This argument yields
the precise claimed inequality \eqref{eq:sum osc int B(s) 4-tuples non-diag}.
A very similar argument to the above (except that we need to deal with $6$-tuples of lattice points rather than $4$-tuples,
and also this time $\Cc_{n}(6)$ replaces $\Dc_{n}(4)$ as, unlike $l=4$ for $l=6$, $\Cc_{n}(6)$ might properly contain $\Dc_{n}(6)$)
also gives \eqref{eq:sum osc int B(s) 6-tuples non-diag}.

\end{proof}

\end{document}